\documentclass[12pt]{amsart}
\usepackage{graphicx}
\usepackage[centertags]{amsmath}
\usepackage{amsfonts}
\usepackage{amssymb}
\usepackage{amsthm}
\usepackage{hyperref}

\newtheorem{thm}{Theorem}
\newtheorem{lem}{Lemma}[section]
\newtheorem{cor}{Corollary}

\newtheorem{prop}[lem]{Proposition}
\theoremstyle{definition}
\newtheorem{defn}[lem]{Definition}
\theoremstyle{remark}
\newtheorem{rem}{Remark}[section]
\numberwithin{equation}{section}

\newcommand{\norm}[1]{\left\Vert#1\right\Vert}

\newcommand{\set}[1]{\left\{#1\right\}}

\newcommand{\pd}[2]{\frac{\partial #1}{\partial #2}}
\newcommand{\pr}[2]{\langle  #1,#2\rangle}

\newcommand{\vphi}{{\varphi}}

\newcommand{\calF}{\mathcal{F}}
\newcommand{\calG}{\mathcal{G}}
\newcommand{\calI}{\mathcal{I}}
\newcommand{\calJ}{\mathcal{J}}
\newcommand{\calK}{\mathcal{K}}

\newcommand{\calP}{\mathcal{P}}
\newcommand{\calS}{\mathcal{S}}

\newcommand{\bbZ}{\mathbb{Z}}

\newcommand{\bbR}{\mathbb R}
\newcommand{\bbC}{\mathbb C}

\newcommand{\bbN}{\mathbb N}

\newcommand{\bbH}{\mathbb{H}}
\newcommand{\G}{\mathcal{G}}
\newcommand{\bfK}{\mathcal{K}}
\newcommand{\bfP}{\mathcal{P}}
\newcommand{\bfL}{\mathbf{L}}
\newcommand{\bfM}{\mathbf{M}}
\newcommand{\Tr}{ \mbox{Tr}}
\newcommand{\Op}{ \operatorname{Op}}

\newcommand{\Sl}{ \mathrm{sl}}
\newcommand{\PSL}{ \mathrm{PSL}}
\newcommand{\PGL}{ \mathrm{PGL}}

\newcommand{\dv}{d\mathrm{vol}}
\newcommand{\vol}{\mathrm{vol}}
\newcommand{\im}{\mathrm{Im}}

\newcommand{\var}{\mathrm{Var}}
\newcommand{\calH}{\mathcal{H}}

\newcommand{\bfE}{\mathbf{E}}
\newcommand{\lap}{\triangle}
\newcommand{\bs}{\backslash}
\newcommand{\id}{1\!\!1}
\renewcommand{\Im}{\mathrm{Im}}

\begin{document}
\title[Quantum Ergodicity for hyperbolic planes]
{Quantum Ergodicity for products of hyperbolic planes}%
\author{Dubi Kelmer }%
\address{Raymond and Beverly Sackler School of Mathematical Sciences,
Tel Aviv University, Tel Aviv 69978, Israel}
\email{kelmerdu@post.tau.ac.il}

\thanks{}%

\date{\today}%
\dedicatory{}%
\commby{}%

\begin{abstract}
    For manifolds with geodesic flow that is ergodic on the unit tangent
    bundle, the quantum ergodicity theorem implies that almost all
    Laplacian eigenfunctions become equidistributed as the
    eigenvalue goes to infinity. For a locally symmetric space with a universal cover that is a product of several upper half
    planes, the geodesic flow has constants of motion so it can not be
    ergodic. It is, however, ergodic when restricted to the submanifolds defined by these constants.
    In accordance, we show that almost all eigenfunctions become
    equidistributed on these submanifolds.
\end{abstract}

\maketitle
\section*{introduction}
    The Quantum Ergodicity Theorem
    ~\cite{ColinDeVerdiere85,Snirelman74,Zelditch87}, is a
    celebrated result concerning the behavior of Laplacian
    eigenfunctions on compact manifolds with an ergodic geodesic
    flow, stating that most eigenfunctions become
    equidistributed on the unit tangent bundle with respect to the volume measure.
    We currently lack a general understanding of the situation when
    the geodesic flow is not ergodic (but still not integrable).
    In this paper we look at a special example, that of a locally
    symmetric space with a universal cover that is a product of upper half planes.
    The geodesic flow on this space is no longer ergodic, yet, it
    does posses some chaotic features, suggesting that Laplacian eigenfunctions still become equidistributed (on the correct space).

    In the special case of one half plane, the geodesic flow is
    ergodic, so that the Quantum Ergodicity Theorem applies.  In fact, in this case
    it is believed that a much stronger result holds, that is, that
    all eigenfunctions become equidistributed as the eigenvalue goes to infinity. This notion is referred to as
    \emph{Quantum Unique Ergodicity} and is conjectured to hold for surfaces of negative curvature \cite{RudSarnak94}.
    Perhaps the strongest evidence for the QUE conjecture comes from the analysis on
    arithmetic surfaces, i.e., $X=\Gamma\bs\bbH$ with $\Gamma$ a congruence subgroup. For
    arithmetic surfaces there are additional symmetries, Hecke operators commuting
    with each other and with the Laplacian. A joint eigenfunction of the Laplacian and all Hecke operators
    is called a Hecke eigenfunction. In \cite{Linden06}, Lindenstrauss showed that indeed, for any sequence of
    Hecke eigenfunctions the corresponding quantum measures converge to the volume measure.
    We note that by Watson's formula for triple integrals \cite{Watson01},
    the Grand Riemann Hypothesis implies QUE for Hecke eigenfunctions
    with an effective rate of convergence.

    We now proceed to the high rank case and consider the locally symmetric space $X=\Gamma\bs\calH$, with
    $\calH=\bbH\times\cdots\times\bbH$ a product of $d$ hyperbolic planes,
    and $\Gamma$ an irreducible co-compact lattice in $\calG=\PSL(2,\bbR)^d$.
    We note that in this case the geodesic flow (on the tangent bundle) has $d>1$ independent constants of
    motion given by the partial energy functions
    $E_j(z,\xi)=\norm{\xi_j}^2_{z_j}$. Consequently, the geodesic flow can not be
    ergodic on the unit tangent bundle. It is, however, ergodic when restricted to the generalized energy
    shells $$\Sigma(\bfE)=\set{(z,\xi)\in TX|
    E_j(z,\xi)=E_j}\subset SX,$$ for any level $\bfE\in [0,1]^d$ with $\sum_j E_j=1$
    (we normalize the energy so that the energy shell lies in $SX$).
    Notice, that the structure and the dynamics on each energy shell
    is determined by the set of singularities $\set{j|E_j=0}$, so that any two energy
    shells with the same singularities can be identified.

    The algebra of invariant differential operators
    in this case is generated by $d$ partial Laplacians $\lap_j=y_j^2(\pd{^2}{x_j^2}+\pd{^2}{y_j^2})$ acting on each hyperbolic
    plane. Any Laplacian eigenfunction $\phi_k$ (that is, a joint eigenfunction of all the partial Laplacians $\lap_j\phi_k+\lambda_{k,j}\phi_k=0$),
    can be interpreted as a distribution on the unit tangent bundle $SX$ (via a corresponding Wigner distribution).
    This distribution is concentrated on a corresponding energy shell $\Sigma(\bfE_k)\subseteq SX$,
    with $\bfE_k=\frac{(\lambda_{k,1},\ldots,\lambda_{k,d})}{\lambda_{k,1}+\cdots+\lambda_{k,d}}$.
    The projection of this distribution to the base manifold $X$ is the measure defined by the density $d\mu_k=|\phi_k|^2dz$.

    In this paper we show an analogous result to the
    Quantum Ergodicity Theorem in this setting.
    For every energy level $\bfE$ the flow on the energy shell $\Sigma(\bfE)$ is
    ergodic. Correspondingly, we show that for almost any sequence of eigenfunctions, $\phi_k$,
    with (normalized)
    eigenvalues
    $\frac{(\lambda_{k,1},\ldots,\lambda_{k,d})}{\lambda_{k,1}+\cdots+\lambda_{k,d}}=\bfE_k\to\bfE$,
    the corresponding distributions converge to the volume measure of the energy shell $\Sigma(\bfE)$.
    Note that the projection to the base of the volume measure from any energy shell is the volume measure of $X$,
    hence, the above result implies that for almost any sequence of eigenfunctions,
    the measures $\mu_k$ converge to the volume measure of $X$.

\begin{rem}
It is reasonable that similar results could be proved more generally with pseudodifferential calculus. 
See \cite{Zelditch92} for analogous results on reduced quantum ergodicity in the presence of symmetries.
\end{rem}

    The above Quantum Ergodicity result holds for any irreducible co-compact lattice and for any orthonormal basis of
    Laplacian eigenfunctions (without taking into account action of Hecke operators).
    We note that, as in rank one, much more is known in the arithmetic setting. When
    $\Gamma$ is a congruence subgroup (coming from a Quaternion
    algebra over a corresponding number field), there are Hecke
    operators acting on $L^2(\Gamma\bs\calH)$ commuting with all the partial Laplacians.
    If one considers Hecke eigenfunctions then it is likely that, again, the only limiting measure on $X$ obtained as a quantum limit
    (respectively its lift to $\Gamma\bs\calG$), is the volume
    measure\footnote{As pointed out by Elon Lindenstrauss, such a result should follow from the same arguments applied in
    \cite{BourgainLinden03,Linden01,Linden06}.}.
    In particular, this would imply that for \textbf{any} sequence of Hecke
    eigenfunctions, with eigenvalues $\bfE_k\to\bfE$, the
    corresponding distributions converge to the volume measure of
    $\Sigma(\bfE)$.

    \begin{rem}
    In \cite{SilbermanVenkatesh04I}, Silberman and Venkatesh generalized
    the lift of the limiting measures to the more general setting of higher rank locally symmetric spaces $\Gamma\bs\calG/K$
    with $\calG$ a semi-simple connected Lie group.
    In \cite{SilbermanVenkatesh04II}, for
    the special case of $\calG=\PGL(d,\bbR)$ with $d$ prime, they used this lift to
    generalize the results of \cite{Linden06},
    and show that for any sequence of (non-degenerate)
    Hecke eigenfunctions, the limiting measure is the Haar measure.
    \end{rem}

\subsection*{Results}
Let $X=\Gamma\bs\calH$, with
$\calH=\bbH\times\cdots\times\bbH$ a product of $d$ hyperbolic planes,
and $\Gamma$ an irreducible co-compact lattice in $\calG=\PSL(2,\bbR)^d$.
Let $\{\phi_k\}$ be an orthonormal basis of $L^2(X)$ consisting
of (joint) Laplacian eigenfunctions $(\lap_j+\lambda_{k,j})\phi_k=0,\;\lambda_{k,j}=\frac{1}{4}+r_{k,j}^2$.
For any $\bfL=(L_1,\ldots, L_d)\in[\frac{1}{2},\infty)^{d}$, let
$$\calI(\bfL)=\set{k\colon \norm{r_{k}-\bfL}_\infty\leq 1/2}$$
denote the set of eigenfunctions with eigenvalues in a window around
$\bfL$, and denote by $N(\bfL)=\sharp\calI(\bfL)$ the number of such
eigenfunctions.
\begin{rem}
The choice for the window, $\calI(\bfL)$, to be of volume one in $\bbR^d$ is mainly cosmetic. The same results (with essentially the same proofs) also holds if we take $\calI(\bfL)$ to be a window of any given size.
\end{rem}

For any $\bfE\in[0,1]^d$ (with $\sum E_j=1$) we can identify the generalized energy shells, $\Sigma(\bfE)$,
with the quotients $\Gamma\bs
\calG/\prod_{E_j=0}K_j$. So for instance when $d=2$ and $\bfE=(1,0)$ the energy shell $\Sigma(1,0)$ is identified
with $\Gamma\bs \PSL(2,\bbR)\times\bbH$.
Under this identification the volume measure of $\Sigma(\bfE)$ is (up to normalization)
the measure induced from the Haar measure of $\calG$.

To each eigenfunction $\phi_k$ we attach a distribution $S_{\phi_k}$ on
$\Gamma\bs\calG$  (coming from the Wigner distribution on $TX$),
that coincides with $\mu_k$ on $\calK$-invariant function (see sections \ref{sLift}).
If we take $\bfL\to\infty$ so that $\frac{(L_1^2,\ldots,L_d^2)}{\norm{\bfL}^2}\to
\bfE$, then for $k\in\calI(\bfL)$ the normalized eigenvalues
$\bfE_k\sim\bfE$ and the distributions $S_{\phi_k}$ become close to probability measures on $\Sigma(\bfE)$.
Given a smooth test function $a\in C^\infty(\Sigma(\bfE))$,
we evaluate how far are these distributions (equivalently measures) from the volume measure.
We first show that when (at least one of) the eigenvalues go to infinity, on average, the distributions $S_{\phi_k}$ converge to the volume measure on $\Gamma\bs \calG$.
\begin{thm}\label{tWL}
    For any $a\in C^\infty(\Gamma\bs \calG)$,
    \[\lim_{\norm{\bfL}\to\infty}\frac{1}{N(\bfL)}\sum_{k\in\calI(L)}S_{\phi_k}(a)=\frac{1}{\vol(\Gamma\bs \calG)}\int_{\Gamma\bs \calG}a(g)dg.\]
\end{thm}
Note that this theorem holds for any smooth function on $\Gamma\bs \calG$, and any limit $\bfL\to\infty$. In particular, if we start from $a\in C^\infty(\Sigma(\bfE))$ and take $\bfL\to\infty$ with $\frac{(L_1^2,\ldots,L_d^2)}{\norm{\bfL}^2}\to
\bfE$, we get that on average 
$$\frac{1}{N(\bfL)}\sum_{k\in\calI(L)}S_{\phi_k}(a)\to \int_{\Sigma(\bfE)}a\dv.$$

Next, we study the variation from the average. To do this, define the
variance of the distributions $S_{\phi_k},\;k\in\calI(\bfL)$ (with respect to the test
function $a\in C^\infty(\Sigma(\bfE)$) as
\[\var_\bfL(a)=\frac{1}{N(\bfL)}\sum_{k\in\calI(\bfL)}\left|S_{\phi_k}(a)-\int_{\Sigma(\bfE)}a\dv\right|^2.\]
\begin{thm}\label{tVAR2}
For any $a\in C^\infty(\Sigma(\bfE))$, as $\bfL\to\infty$ with
$\frac{(L_1^2,\ldots,L_d^2)}{\norm{\bfL}^2}\to \bfE$,
$$\lim_{\norm{\bfL}\to\infty}\var_\bfL(a)=0.$$
\end{thm}

In particular, using a diagonalization argument (see \cite{Zelditch87}),
this variance estimate implies that for almost any sequence of $\phi_k$,
with normalized eigenvalues converging to some energy $\bfE$,
the distributions $S_{\phi_k}$ converge to the volume measure of
the corresponding energy shell.

Moreover, since the projection of the volume measure of any
energy shell to $X$ is the volume measure of $X$, we get the following
corollary:
\begin{cor}
For any $a\in C^\infty(X)$, as $\bfL\to\infty$
$$\lim_{\norm{\bfL}\to\infty}\frac{1}{N(\bfL)}\sum_{k\in\calI(\bfL)}\left|\int_X a(z)|\phi_k(z)|^2dz-\int_{X}a(z)dz\right|^2=0.$$
\end{cor}
This implies that almost all of the measures $\mu_k$ converge
to the volume measure of $X$.

\section*{Acknowledgments}
    I would like to thank Lior Silberman for explaining his results on the micro local lift for locally symmetric spaces.
    I also thank Elon Lindenstrauss for his remarks regarding arithmetic
    quantum unique ergodicity. I thank Mikhail Sodin for his helpful suggestions. Finally, I thank Zeev Rudnick for sharing his
    insights and for his comments on an early draft of this paper. This work was supported in part by the Israel Science Foundation founded by the Israel Academy of Sciences and Humanities.

\section{Background and Notation}\label{sNote}
\subsection{The hyperbolic plane}
Let $\bbH=\set{z\in\bbC\colon \im(z)>0}$ denote the upper half plane.
This is a symmetric space with group of isometries
$G=\PSL(2,\bbR)$ (acting by linear transformations). Let $K\subseteq G$ be the stabilizer of
$i\in\bbH$ and $P\subseteq G$ be the stabilizer of
$\infty\in\partial\bbH$. We use coordinates corresponding to the identification
of $\bbH=G/K\cong P$. For $z=x+iy\in\bbH$ and $\theta\in[0,\pi)$ we let
$p_{z}=\begin{pmatrix}
\sqrt{y}  & x/ \sqrt{y}  \\
0 & 1/ \sqrt{y}
\end{pmatrix}=\begin{pmatrix}
1 & x \\
0 & 1\end{pmatrix}\begin{pmatrix}
\sqrt{y} & 0 \\
0 & \sqrt{1/y}    \end{pmatrix}
\in P$ and $k_{\theta}=\begin{pmatrix} \cos(\theta)& \sin(\theta)\\
-\sin(\theta) & \cos(\theta) \end{pmatrix}\in K$. In these coordinates the (normalized) Haar measures of
$P$ and $K$ are given by $dp=dz=\frac{dxdy}{y^2}$ and
$dk=\frac{d\theta}{\pi}$. The Haar measure of $G$ is then
$dg=dpdk$.

The Lie algebra $\Sl(2,\bbR)$, is generated by
$W=\begin{pmatrix}0 & 1\\-1&
0\end{pmatrix},\;H=\begin{pmatrix} 1 & 0\\ 0& -1\end{pmatrix}$
and $X=\begin{pmatrix} 0 & 1\\ 0 & 0\end{pmatrix}$.
In the $(x,y,\theta)$ coordinates these elements are given
by the following differential operators (c.f. \cite[Chapter  IV \S 4]{Lang85}):
\begin{eqnarray*}
W&=&\pd{}{\theta}\\
H&=&-2y\sin(2\theta)\pd{}{x}+2y\cos(2\theta)\pd{}{y}+\sin(2\theta)\pd{}{\theta}\\
X&=&y\cos(2\theta)\pd{}{x}+y\sin(2\theta)\pd{}{y}+\sin^2(\theta)\pd{}{\theta}\\
\end{eqnarray*}

We consider the function $E(z,\xi)=\norm{\xi}^2_z=\frac{\xi_x^2+\xi_y^2}{y^2}$ as an energy function. For any $E>0$ we can identify
the energy shell $$T\bbH_E=\set{(z,\xi)\colon E(z,\xi)=E}$$ with $G$ (and also with the unit tangent bundle $T\bbH_1$) via the map $g(z,\xi)=p_zk_\theta$ with
$\theta=\tan^{-1}(\xi_y/\xi_x)$. The zero section $T\bbH_0\cong\bbH$ is identified with $G/K$.
The energy shells are invariant under the geodesic flow, that (under the above identification) is given by the action of
$A(t)=e^{Ht}=\begin{pmatrix}e^t & 0\\ 0 & e^{-t}\end{pmatrix}$
on $G$ from the right.

\subsection{Products of hyperbolic planes}
Let $\calH=\bbH\times\cdots\times\bbH$ be a product of $d$ hyperbolic
planes. This is a symmetric space with group of isometries $\G=\prod_{j=1}^d G_j$ (where $G_j=G=\PSL(2,\bbR)$).
We have a decomposition $\G=\bfP\bfK$, where
$\bfK=\prod K_j$ and $\bfP=\prod P_j$. For $z=(z_1,\ldots, z_d)\in\calH$ and $\theta=(\theta_1,\ldots,\theta_d)\in [0,\pi)^d$
we let $p_z=(p_{z_1},\ldots ,p_{z_d})\in \calP$ and
$k_\theta=(k_{\theta_1},\ldots,k_{\theta_d})\in \calK$. The Haar measures of $\bfP,\bfK$, and
$\G$ are then $dp=dp_1\cdots dp_d,\;dk=dk_1\cdots dk_d$ and $dg=dg_1\cdots dg_d$ respectively.
We will denote by $W_j,H_j$ and $X_j$ the action of the
corresponding differential operator on the $j$'th factor. 
For every $\bfE=(E_1,\ldots,E_d)\in [0,\infty)^d$ we can identify $\calG/\prod_{E_j=0}K_j$ 
with a corresponding energy shell inside the tangent bundle $T\calH$. The geodesic flow through the point
$(g,\bfE)=((g_1,\ldots,g_d),(E_1,\ldots, E_d))\in \calG\times
[0,\infty)^d$ is then given by the
right action of $A_\bfE(t)=\prod_j A_j(\sqrt{E}_jt)$ (where each
$A_j(t)=e^{tH_j}$ is the diagonal action on the corresponding
factor).

\subsection{Irreducible lattices}
    A discrete subgroup $\Gamma\subset \calG$ is called a lattice if the quotient $\Gamma\bs \calG$
    has finite volume, and co-compact when $\Gamma\bs \calG$ is
    compact. We say that a lattice $\Gamma\subset\calG$ is
    irreducible, if for every (non-central) normal subgroup $N\subset \calG$ the projection of $\Gamma$ to $\calG/N$ is dense.
    An equivalent condition for irreducibility, is that for any nontrivial
    $1\neq\gamma\in\Gamma$, none of the projections $\gamma_j\in G_j$ are trivial \cite[Theorem 2]{Shimizu}.
    Examples of irreducible (co-compact) lattices can be constructed from
    norm one elements of orders inside an appropriate quaternion algebra over a totally real number field \cite{Shavel76}.
    In fact when $d\geq 2$, Margulis's arithmeticity theorem states that, up to commensurability, these are the only examples.

    Let $\Gamma\subset \G$ be an irreducible co-compact lattice,
    we now go over the classification of the different elements of $\Gamma$ (see e.g., \cite{Efrat87,Shimizu}).
    Recall that an element $g_j\in G_j=\PSL(2,\bbR)$ is called hyperbolic if $|\Tr(g_j)|>2$, elliptic if $|\Tr(g_j)|<2$,
    and parabolic if $|\Tr(g_j)|=2$. Then for any nontrivial $1\neq\gamma\in\Gamma$, the projections to the different factors are either
    hyperbolic or elliptic. There are purely hyperbolic elements (where all projections are hyperbolic),
    and mixed elements (where some projections are hyperbolic and other
    elliptic). There could also be a finite number of torsion points that are purely elliptic.

\subsection{Spectral decomposition}
    Let $\Gamma$ be an irreducible co-compact lattice in $\calG$,
    so that $X=\Gamma\bs\calH$ is a compact Riemannian
    manifold. The algebra of invariant differential operators is generated by the $d$ partial hyperbolic
    Laplacians $\lap_j=y_j^2(\pd{^2}{x_j^2}+\pd{^2}{y_j^2})$. The
    space $L^2(X)$ has an orthonormal basis, $\{\phi_k\}$, composed of Laplacian
    eigenfunctions, $(\lap_j+\lambda_{k,j})\phi_k=0$,
    with all partial eigenvalues $\lambda_{k,j}\geq 0$ (in fact $\lambda_{k,j}>0$ unless $\phi_k\equiv 1$).
    For each eigenfunction $\phi_k$, we think of the normalized eigenvalues
    $\bfE_k=\frac{(\lambda_{k,1},\ldots,\lambda_{k,d})}{\lambda_{k,1}+\cdots+\lambda_{k,d}}$ as a quantum energy level.
    We use the standard parametrization
    $\lambda_{k,j}=r_{k,j}^2+\frac{1}{4}$ with $r_{k,j}\geq 0$ for $\lambda_{k,j}\geq \frac{1}{4}$ and $r_{k,j}\in i(0,\frac{1}{2})$ otherwise.

    In the case where some of the eigenvalues are small, $\lambda_{k,j}<\frac{1}{4}$, we say that $\phi_k$ is
    exceptional. We note that in this setting there could be infinitely many exceptional eigenfunctions
    (in contrast to the rank one case where there could be only finitely many).
    When $\Gamma$ is a congruence subgroup, it is
    conjectured that there are no exceptional eigenfunctions at all.
    However, in our general setting the most we can say is that
    the exceptional eigenfunctions are of density zero (Lemma \ref{lSmall1}).

\subsection{Fourier decomposition}
    Consider the homogeneous space $Y=\Gamma\bs\calG$.
    We can identify our space $X=\Gamma\bs \calH$ with
    the double quotient $\Gamma\bs \G/\bfK=Y/\calK$, and think of functions on $X$ as
    $\calK$-invariant functions on $Y$. For $n\in\bbZ^d$, let $\calF_{n}(Y)$ denote the joint eigenspaces
    for $W_j$ with eigenvalues $2in_j$ respectively. That is
    \[\calF_n(Y)=\set{f\in C^\infty(Y)| f(p_zk_\theta)=e^{2i(n\cdot \theta)}f(p_z)}.\]
    Any $a\in C^\infty(Y)$ has a $\calK$-Fourier
    decomposition $a=\sum_{n\in\bbZ^d}a_{n}$ with $a_n\in\calF_n(Y)$.
    For any integer $s\geq 1$, the functions $a_n$ in this decomposition are uniformly bounded $\norm{a_{n}}_\infty\leq\frac{
    \norm{W_j^sa}_\infty}{|2n_j|^{s}}$ for any $1\leq j\leq d$.

    Consider the operators $E_j^\pm=H_j\pm i(2X_j-W_j)$. These operators satisfy
    $[W_j,E_j^\pm]=2iE_j^\pm$, so that they act as raising an lowering
    operators (i.e.,
    $E_j^\pm:\calF_{n}(Y)\to\calF_{n\pm e_j}(Y)$). Furthermore, on
    $K_j$ invariant functions the action of $E_j^-E_j^+$ coincides with
    the action of ($4$ times) the $j$'th partial Laplacian.

\subsection{Reduced Ergodicity}
The geodesic flow on $TX$ is the flow induced from the geodesic flow on $T\calH$.
The unit tangent bundle,
\[SX=\set{(z,\xi)\in X| \sum E_j(z_j,\xi_j)=1},\]
is invariant under this flow. However, in contrast to the rank one
case, the flow on $SX$ is no longer ergodic (because the functions $E_j(z_j,\xi_j)$ are $d$ independent constants of
motion).
Instead of the unit tangent bundle, for any $\bfE=(E_1,\ldots, E_d)\in [0,\infty)^d$ we consider a generalized energy shell
\[\Sigma(\bfE)=\set{(z,\xi)\in TX|\forall j,\;E_j(z_j,\xi_j)=E_j}.\]
We can naturally identify these energy shells
$\Sigma(\bfE)\cong \Gamma\bs\calG/\prod_{E_j=0}K_j$,
and think of functions on $\Sigma(\bfE)$ as $\prod_{E_j=0}K_j$ invariant functions on $Y$.
By Moore's ergodicity theorem \cite[Theorem 2.2.6]{Zimmer84},
the geodesic flow restricted to $\Sigma(\bfE)$ is ergodic (with
respect to the volume measure on $\Sigma(\bfE)$).
In fact, each one of the flows $A_j(t)=e^{H_jt}$ is already ergodic on $Y$.
In particular, for any $a\in C^\infty(Y)$ let
\[\langle a\rangle_\bfE^T=\frac{1}{T}\int_0^T a\circ A_\bfE(t)dt,\]
denote the time average of $a$ with respect to the geodesic flow $A_{\bfE}(t)=\prod_j A_j(\sqrt{E_j}t)$ on $Y$.
Then as $T\to\infty$, this time averages converges, in
$L^2(Y)$, to the phase space average
$$\langle a\rangle_\bfE^T\stackrel{L^2(Y)}{\longrightarrow}\int_{Y} a(g)dg.$$
Note that this property is stronger then ergodicity on $\Sigma(\bfE)$, as we do not assume that $a$ is $K_j$ invariant when $E_j=0$.

\subsection{Notations}
    We make use of the following notation: Given a positive function $g(x)$, we denote $f(x)=O(g(x))$
    or $f(x)\ll g(x)$ if there exists a constant $c>0$ such that
    $\forall x,\;|f(x)|\leq cg(x)$. If the constant depends on some
    parameters, say $\epsilon,\delta$, then they will appear as a subscript
    $f(x)=O_{\epsilon,\delta}(g(x))$ or $f(x)\ll_{\epsilon,\delta}
    g(x)$. We also use the notation $f(x)=o(g(x))$ meaning that
    $\lim_{x\to\infty}\frac{|f(x)|}{g(x)}=0$.
\section{Outline of proof}
    We now describe the outline for the proof of the main result (Theorem \ref{tVAR2}).
    For each eigenfunction
    $\phi_k$, we identify the corresponding Wigner distribution with a distribution $S_{\phi_k}$ on $Y=\Gamma\bs\calG$.
    On $\calK$-invariant functions $a\in C^\infty(\Gamma\bs
    \calG/\calK)=C^\infty(X)$ these distributions coincide with the
    quantum measures,
    \[S_{\phi_k}(a)=\int_X a(z)|\phi_k(z)|^2dz.\]
    We will show that the distributions $S_{\phi_k}$ satisfy the following
    properties:\\
    \textbf{1.(Invariance)}
     For $j=1,\ldots, d$ consider the ergodic flows $A_j(t)$ on $Y$
     (given by the right action of
     $\begin{pmatrix} e^{t} & 0\\ 0 & e^{-t}\end{pmatrix}$
     on the $j$'th factor). Then, as
     $r_{k,j}\to\infty$ the distributions $S_{\phi_k}$ becomes invariant
     in the sense that for $a\in C^\infty(Y)$,
    \[|S_k(a\circ A_j(t))-S_k(a)|\ll_{a,t} \frac{1}{r_{k,j}}.\]
    \textbf{2.(Positivity)}
    When the normalized eigenvalue $\bfE_k$ is close to an energy level $\bfE$, the
    distribution $S_{\phi_k}$ is close to positive measure on
    $\Sigma(\bfE)$. More precisely, for every eigenfunction $\phi_k$, and any $\calJ\subseteq\{1,\ldots,d\}$
    there is normalized function $\Phi_{k,\calJ}\in L^2(Y)$,
    such that for functions $a\in C^\infty(\Gamma\bs\calG/\prod_{j\not\in \calJ}K_{j})$
    \[S_{\phi_k}(a)=\pr{a\Phi_{k,\calJ}}{\Phi_{k,\calJ}}+O_{a}(R_{k,\calJ}^{-1/4}),\]
    with $R_{k,\calJ}=\min_{j\in\calJ}r_{k,j}$.\\
    \textbf{3.(Local Weyl's law)} The last property, stated in Theorem \ref{tWL}, deals with the average of the distributions $S_{\phi_k}$
    over the window $\calI(\bfL)$. That is, for $a\in C^\infty(Y)$,
    \[\lim_{\norm{\bfL}\to\infty}\frac{1}{N(\bfL)}\sum_{k\in\calI(\bfL)}S_k(a)=\frac{1}{\vol(Y)}\int_Y a(g)dg .\]

    Assuming these properties are satisfied the proof goes as follows:
    \begin{proof}[\textbf{Proof of Theorem \ref{tVAR2}}]
    Fix $\bfE\in[0,1)^d$ (with $\sum_j E_j=1$) and assume that $\bfL\to\infty$ with
    $\frac{(L_1^2,\ldots,L_d^2)}{\norm{\bfL}^2}\to\bfE$. If we set $\calJ=\set{j|E_j\neq
    0}$, then for any $j\in\calJ$ we have $L_j\gg\norm{\bfL}$.
    Fix a smooth function $a\in C^\infty(\Sigma(\bfE))$ and assume $\int_{\Sigma(\bfE)}a\dv=0$.
    Identify $\Sigma(\bfE)=\Gamma\bs \calG/\prod_{i\not\in \calJ}K_i$, and
    think of $a$ as a function in $C^\infty(Y)$ invariant under $\prod_{i\not\in\calJ}K_i$.
    Now consider the ``time average"
    with respect to the geodesic flow $A_\bfE(t)=\prod_{j\in\calJ}A(\sqrt{E}_jt)$,
    \[\langle a\rangle_\bfE^T=\frac{1}{T}\int_0^T a\circ A_\bfE(t)dt.\]
    For any $k\in\calI(\bfL)$ and $j\in\calJ$, the distribution $S_{\phi_k}$ is $A_j(t)$ invariant (up to $O(\frac{1}{\norm{\bfL}})$).
    Hence, for $k\in\calI(\bfL)$
    \[S_{\phi_k}(\langle
    a\rangle_\bfE^T)=S_{\phi_k}(a)+O_{a,T}(\frac{1}{\norm{\bfL}}).\]
    Consequently,
    \[\frac{1}{N(\bfL)}\sum_{k\in\calI(\bfL)}|S_{\phi_k}(a)|^2
    = \frac{1}{N(\bfL)}\sum_{k\in\calI(\bfL)}|S_{\phi_k}(\langle
    a\rangle_j^T)|^2+O_{a,T}(\frac{1}{\norm{\bfL}}).\]

    Since the action of $A_{j}(t),\;j\in\calJ$ commutes with the action of $K_{i}$ for any $i\not\in\calJ$ the time average $\langle
    a\rangle_\bfE^T$ remains invariant under $\prod_{i\not\in\calJ}K_{i}$.
    We can use this invariance to show for any $k\in\calI(\bfL)$
    $$|S_{\phi_k}(\langle
    a\rangle_\bfE^T)|^2\leq S_{\phi_k}(|\langle
    a\rangle_\bfE^T|^2)+O_{a,T}(\norm{\bfL}^{-1/4}).$$ 
    Indeed, by the positivity property, $$|S_{\phi_k}(\langle a\rangle_\bfE^T)|^2=
    |\pr{\langle a\rangle_\bfE^T\Phi_{k,\calJ}}{\Phi_{k,\calJ}}|^2+O_{a,T}(\norm{\bfL}^{-1/4}),$$
    and by Cauchy-Schwartz
    \[|\pr{\langle a\rangle_\bfE^T\Phi_{k,j}}{\Phi_{k,j}}|^2\leq
   \pr{|\langle
    a\rangle_\bfE^T|^2\Phi_{k,j}}{\Phi_{k,j}}= S_{\phi_k}(|\langle
    a\rangle_\bfE^T|^2)+O_{a,T}(\norm{\bfL}^{-1/4}).\]

    Consequently, for the average also
    \[\frac{1}{N(\bfL)}\sum_{k\in\calI(\bfL)}|S_{\phi_k}(a)|^2\leq
     \frac{1}{N(\bfL)}\sum_{k\in\calI(\bfL)}S_{\phi_k}(|\langle a\rangle_\bfE^T|^2)+O_{a,T}(\norm{\bfL}^{-1/4}).\]
    Taking $\norm{\bfL}\to\infty$ the local Weyl's law then implies
    \[\limsup_{\norm{\bfL}\to\infty}\frac{1}{N(\bfL)}\sum_{k\in\calI(\bfL)}|S_{\phi_k}(a)|^2\leq
    \frac{1}{\vol(Y)}\int_Y|\langle a\rangle_\bfE^T|^2dg.\]
    Finally, the ergodicity of the flows imply that in the limit
    $T\to\infty$,
    \[\langle a\rangle_\bfE^T\stackrel{L^2(Y)}{\longrightarrow}\int_Ya(g)dg=\int_{\Sigma(\bfE)}a\dv=0,\]
    concluding the proof.
    \end{proof}

    It thus remains to verify that the distributions $S_{\phi_k}$ indeed satisfy the desired properties.
    In section \ref{sLift} we follow the arguments used by Lindenstrauss in
    \cite{Linden01} to show the invariance and positivity
    properties. Then in sections \ref{sQuant},\ref{sLW} we will follow Zelditch's formalism~\cite{Zelditch87}
    for the Wigner distribution via the Helgason-Fourier transform to give a local Weyl's law.

\section{Micro Local Lift}\label{sLift}
    In this section we recall the construction of \cite{Linden01,Zelditch87}, lifting a quantum measure $\mu_{\phi_k}$
    on $X$ to a distribution $S_{\phi_k}$ on $Y$. We then verify that
    these distributions satisfy the desired properties of invariance and
    positivity. This is essentially the content of \cite[Theorem 4.1]{Linden01} and \cite[Theorem 3.1]{Linden01},
    however, since the formulation we need is slightly different we will include the proofs.
    Throughout this section the eigenfunction $\phi=\phi_k$ is fixed,
    and for notational convenience the subscript $k$ will be
    omitted.
\subsection{Definition}
    For $r=(r_1,\ldots,r_d)$ let $\phi\in C^\infty(X)\equiv \calF_0(Y)$
    be a joint eigenfunction of all the $\lap_j$'s with
    eigenvalues $\lambda_j=(\frac{1}{4}+r_j^2)$ respectively.
    We construct from $\phi$ by induction a sequence of functions
    $\phi_{n}\in\calF_{n}(Y),\; n\in\bbZ^d$: Let
    $\phi_{0}(g)=\phi(g(i))$, and define
    \begin{equation}\label{ePHInm}
    \phi_{n\pm e_j}=\frac{1}{2ir_j+1\pm
    2n_j}E_j^\pm \phi_{n}
    \end{equation}

    \begin{defn}
    Define distribution $S_{\phi}$ on $C^\infty(Y)$ by
    \[S_{\phi}(a)=\lim_{N\to\infty}\langle
    a \sum_{\norm{n}_\infty\leq N}\phi_{n},\phi_{0}\rangle_Y,\]
    where $\langle
    a,b\rangle_Y=\int_Y a(g)\bar b(g)dg$.
    \end{defn}
     Note that the rapid decay of $\norm{a_n}_\infty$ as $\norm{n}\to\infty$, imply that the sum absolutely
    converges and the distributions $S_{\phi}(a)$ are bounded by
    \[\sum_n \norm{a_n}_\infty\ll\max_j\norm{W_j^{2d}a}_\infty.\]
    \begin{rem}
        This definition coincides with the Wigner distribution
        constructed by Zelditch \cite{Zelditch87} via Helgason's Fourier calculus (see section \ref{sQuant2}).
        Also, see \cite{SilbermanVenkatesh04I} for a representation theoretic interpretation of this construction that is more
        natural when generalizing it to locally symmetric spaces.
    \end{rem}

\subsection{Invariance}
    Recall the family of one parameter (ergodic) flows, $A_j(t)=e^{H_jt},\;1\leq j\leq d$ on $Y$.
    We now show that when the $j$'th eigenvalue $\lambda_j=(\frac{1}{4}+r_j^2)$ becomes large the
    distribution $S_{\phi}$ becomes invariant under the corresponding flow $A_j(t)$
    (c.f., \cite[Theorem 4.1]{Linden01})
    \begin{prop}
    For fixed $a\in C^\infty(Y)$,
    \[|S_{\phi}(a)-S_{\phi}(a\circ A_j(t))|\ll_{a,t}\frac{1}{r_j}\]
    \end{prop}
    \begin{proof}
    The flow $A_j$ is generated by $H_j\in\Sl_2(\bbR)$ in the sense that
    $\frac{d}{dt}(a\circ A_j(t))=H_j(a\circ A_j(t))$.
    Let $F(t)=S_{\phi}(a\circ A_j(t))$, then its derivative is given by
    $F'(t)=S_{\phi}(H_j(a\circ A_j(t))$.

    Now use the differential equation \cite[Proposition 4.2]{Linden01}
    \[S_{\phi}((4ir_jH_j+H_j^2+4X_j^2)a)=0,\]
      to deduce
    \[F'(t)=-\frac{1}{4ir_j}S_{\phi}((H_j^2+4X_j^2)(a\circ A_j(t))).\]
    Let $c_{a,t}=t\sup_{0\leq s\leq t} |S_{\phi}((H_j^2+4X_j^2)(a\circ
    A_j(s)))|$, then
    \[|S_{\phi}(a\circ A_j(t))-S_{\phi}(f)|=|F(t)-F(0)|=|\int_0^tF'(s)ds|\leq
    \frac{c_{a,t}}{4r_j}.\]
    \end{proof}
    \subsection{Positivity}
    Fix a subset $\calJ\subseteq\{1,\ldots,d\}$.
    We show that if for all $j\in\calJ$, $r_j$ becomes
    large the distribution $S_\phi$ is close to a positive measure on $\Gamma\bs
    \calG/\prod_{i\not\in\calJ}K_{i}$ (c.f., \cite[Theorem 3.1]{Linden01}).
    \begin{prop}\label{pPOSITIV}
    There are normalized functions $\Phi_\calJ\in L^2(Y)$,
    such that for any fixed $a\in C^\infty(Y)$
    that is invariant under $\prod_{j\not\in \calJ}K_j$,
    \[S_\phi(a)=\langle a\Phi_{\calJ},\Phi_\calJ\rangle_Y+O_{a}\left(\max_{j\in\calJ}\{r_j^{-1/4}\}\right).\]
    \end{prop}

    For the proof we will use the following lemma
    \begin{lem}\label{l1}
    Let $a\in C^\infty(Y)$, then
    \[\langle a\phi_n,\phi_m\rangle_Y=\langle
    a\phi_{n-e_j},\phi_{m-e_j}\rangle_Y+O(\frac{N\norm{a}_\infty+\norm{E_j^+a}_\infty}{r_j}),\]
    where $N=\max(|n_j|,|m_j|)$
    \end{lem}
    \begin{proof}
    By definition
    \[\langle
    a\phi_n,\phi_m\rangle_Y=\frac{\langle
    aE_j^+\phi_{n-e_j},\phi_{m}\rangle_Y}{(2ir_j+2n_j-1)}.\]
    Replace
    $aE^+\phi_{n-e_j}=E_j^+(a\phi_{n-e_j})-(E_j^+a)\phi_{n-e_j}$, and
    use the bound
    $$|\langle (E^+a)\phi_{n-e_j},\phi_{m}\rangle_Y|\leq
    \norm{E_j^+a}_\infty$$ to get that
    \[\langle
    a\phi_{n},\phi_{m}\rangle_Y=\frac{\langle
    E_j^+(a\phi_{n-e_j}),\phi_{m}\rangle_Y}{(2ir_j+2n_j-1)}+O(\frac{\norm{E_j^+a}_\infty}{r_j}).\]
    Finally notice,
    \begin{eqnarray*}
    \frac{\langle
    E_j^+(a\phi_{n-e_j}),\phi_{m}\rangle_Y}{(2ir_j+2n_j-1)}&=&\frac{\langle
    a\phi_{n-e_j},E_j^-\phi_{m}\rangle_Y}{(2ir_j+2n_j-1)}\\
    &=&\langle
    a\phi_{n-e_j},\phi_{m-e_j}\rangle_Y(1+O(\frac{N}{r_j}))\\
    &=& \langle
    a\phi_{n-e_j},\phi_{m-e_j}\rangle_Y+O(\frac{N\norm{a}_\infty}{r_j}).
    \end{eqnarray*}
    \end{proof}

    For any subset $\calJ$, let
    $\bbZ^\calJ=\set{n\in\bbZ^d|\forall j\not\in\calJ,\;n_j=0}$ and for
    any positive integer $N$ let $Z^\calJ_N=\set{n\in \bbZ^\calJ|\norm{n}_\infty\leq N}$.
    Define the function
    $$\Phi_{N,\calJ}=(\frac{1}{2N+1})^{\frac{|\calJ|}{2}}\sum_{n\in Z^\calJ_N}\phi_{n}.$$
    \begin{prop}
    Fix a subset $\calJ\subseteq\{1,\ldots,d\}$ with $|\calJ|=J>0$.
    Let $a\in C^\infty(Y)$ be invariant under $\prod_{j\not\in \calJ}K_j$, and let $R=\min_{j\in\calJ} r_j$. Then
    \[S_{\phi}(a)=\langle
    a\Phi_{N,\calJ},\Phi_{N,\calJ}\rangle_Y+O_{a}(\frac{N^2}{R})+O_{a}(\frac{R^{J\epsilon}}{N^J})+O_{a,\epsilon}(\frac{1}{R}),\]
    \end{prop}
    (Taking $N\sim R^{1/3}$ and $\epsilon=\frac{1}{3}-\frac{1}{4J}$ gives the result of
    Proposition \ref{pPOSITIV}.)
    \begin{proof}
    Since $a$ is invariant under $\prod_{j\not\in\calJ}K_j$, its $\calK$-Fourier decomposition
    is of the form $a=\sum_{n\in \bbZ^\calJ}a_{n}$. Let
    $a_\epsilon=\sum_{n\in Z^{\calJ}_{R^\epsilon}}a_{n}$, then
    $S_j(a)=S_j(a_\epsilon)+O_{a,\epsilon}(\frac{1}{R})$ and $\langle
    a\Phi_{N,\calJ},\Phi_{N,\calJ}\rangle_Y=\langle
    a_\epsilon\Phi_{N,\calJ},\Phi_{N,\calJ}\rangle_Y+O_{a,\epsilon}(\frac{1}{R})$, so it
    is sufficient to prove this for $a_\epsilon$.

    By repeating Lemma \ref{l1} at most $N$ times for each $j\in\calJ$, we get
    \[\langle
    a_\epsilon\Phi_{N,\calJ},\Phi_{N,\calJ}\rangle_Y=(\frac{1}{2N+1})^{J}\sum_{n,m\in Z^\calJ_N}\langle
    a_\epsilon\phi_{(n-m)},\phi_0\rangle_Y+O_a(\frac{N^2}{R})=\]
    \[=(\frac{1}{2N+1})^{J}\sum_{n\in Z^\calJ_N}\sum_{n+m\in Z^\calJ_N}\langle
    a_\epsilon\phi_{m},\phi_0\rangle_Y+O_a(\frac{N^2}{R}).\]
    Now note that $\langle a_\epsilon\phi_{m},\phi_0\rangle_Y=0$ unless $m\in
    Z^\calJ_{R^\epsilon}$. Consequently,
    \begin{eqnarray*}
    \lefteqn{\langle a_\epsilon\Phi_{N,\calJ},\Phi_{N,\calJ}\rangle_Y=}\\
    &=&(\frac{1}{2N+1})^{J}\sum_{m\in
    Z^\calJ_{R^\epsilon}}\langle a_\epsilon\phi_{m},\phi_0\rangle_Y\sharp\set{n,n+m\in Z^\calJ_N}+O_a(\frac{N^2}{R})
    \end{eqnarray*}
    and
    \[S_\phi(a_\epsilon)=(\frac{1}{2N+1})^J\sum_{m\in Z^\calJ_{R^\epsilon}}\langle a_\epsilon\phi_m,\phi_0\rangle_Y\sharp\set{n\in Z^\calJ_N},\]
    We can thus bound the difference
     \begin{eqnarray*}
    \lefteqn{|S_\phi(a_\epsilon)-\langle a_\epsilon\Phi_{N,\calJ},\Phi_{N,\calJ}\rangle_Y|\ll_a}\\
        &\ll_a& \left(\frac{1}{2N+1}\sum_{|m_j|\leq
        R^\epsilon}\sharp\set{n_j\colon |n_j|\leq
        N<|n_j+m_j|}\right)^{J}+O_a(\frac{N^2}{R})\\
        &=& O_a(\frac{R^{J\epsilon}}{N^J})+O_a(\frac{N^2}{R}).
    \end{eqnarray*}
    \end{proof}

\section{Quantization procedure}\label{sQuant}
We now wish to relate the micro local lift defined above, to the
lift obtained via a quantization procedure. That is, for smooth
functions $a\in C^\infty(TX)$ we assign operators $\Op(a)$ on
$L^2(X)$, and for any Laplacian eigenfunction $\phi_k$, we assign
the distribution $a\mapsto \langle \Op(a)\phi_k,\phi_k\rangle$. We
show that this functional is supported on $\Sigma(\lambda_k)$ and
that after identifying $\Sigma(\lambda_k)\cong\Sigma(\bfE_k)\cong
\Gamma\bs\calG$ this functional coincides with the functional $S_{\phi_k}$
defined above.
\subsection{Spherical Transforms}\label{sSpherical}
    Before proceeding with the construction, we digress and go over some of
    Helgason's results on hyperbolic harmonic analysis on
    $\PSL(2,\bbR)$ that we will need \cite{Helgason81}.
    In particular we will make use of the generalized spherical
    functions and spherical transforms. For the rest of this section
    we will concentrate on a single factor $G_j=\PSL(2,\bbR)$, and
    for notational convenience the subscript $j$ will be omitted.

    For $n\in\bbZ$ let $\chi_n$ be the character of $K$ given by $\chi_n(k_\theta)=e^{2in\theta}$ and complete it to a function on $G$ by
    $\chi_n(pk)\equiv \chi_n(k)$. The generalized spherical
    functions $\Phi_{r,n}\in C^\infty (\bbH)$ are given by
    \[\Phi_{r,n}(z)=\int_K \vphi_r(k^{-1}z)\chi_n(k)dk,\]
    where $\vphi_r$ is the Laplacian eigenfunction
    $\vphi_r(x+iy)=y^{ir+\frac{1}{2}}$. Note that both $\Phi_{r,n}$ and $\Phi_{-r,n}$ are Laplacian
    eigenfunctions (with the same eigenvalue $\lambda=r^2+\frac{1}{4}$) and they both satisfy
    $\Phi_{\pm r,n}(kz)=\chi_n(k)\Phi_{\pm r,n}(z)$. Therefore, $\Phi_{r,n}$ and $\Phi_{-r,n}$ differ by some
    constant, which can be computed explicitly as a quotient of $\Gamma$ functions \cite[Proposition 4.17]{Helgason81}
    \[\Phi_{r,n}(z)=\frac{P_n(2ir)}{P_n(-2ir)}\Phi_{-r,n}(z),\]
    with $P_n(x)=\frac{\Gamma(\frac{x+1}{2}+|n|)}{\Gamma(\frac{x+1}{2})}=(\frac{x+1}{2})(\frac{x+1}{2}+1)\cdots(\frac{x+1}{2}+|n|-1)$.
\begin{rem}
For the interested reader, we remark that the spherical function, $\Phi_{r,n}$, can be expressed as a product of $\Gamma$ functions and the $|n|$'th order Legandre function \cite[Proposition 2.9]{Zelditch88}
\[\Phi_{r,n}(ie^t)=\frac{\Gamma(ir+\frac{1}{2}-|n|)}{\Gamma(ir+\frac{1}{2})}P_{ir-\frac{1}{2}}^{|n|}(\cosh(t)).\]
See also \cite[equation 59]{Helgason81} for another expression involving the hypergeometric function.
\end{rem}

    We will not make any direct use of these formulas, all we will use is the following asymptotic estimate on the spherical functions.
    \begin{lem}\label{lBoundSPherical}
    As $r\to\infty$
    \[|\Phi_{r,n}(z)|\ll \frac{1}{\sqrt{r}},\]
    uniformly in any compact set not containing $i$.
    \end{lem}
    \begin{proof}
    Since $|\Phi_{r,n}(kz)|=|\Phi_{r,n}(z)|$, it is sufficient to
    show the bound for $\Phi_{r,n}(iy)$ for $y\neq 1$.

    We can write, $\vphi_r(k_\theta(iy))=e^{(2ir+1)\psi(y,\theta)}$,
    with
    \[\psi(y,\theta)=\frac{1}{2}\log(\frac{y}{\sin^2(\theta)(y^2-1)+1}).\]

    Now, for fixed $y\neq 1$ the function $\psi_y(\theta)=\psi(y,\theta)$ is a smooth function, its first derivative
    $\psi_y'(\theta)$ vanishes only when $\theta=\frac{\pi l}{2},\;l\in\bbZ$
    and the second derivative $\psi_y''(\frac{\pi l}{2})=1-y^{\pm 2}\neq 0$ do not vanish at these points.
    We can now write
    $\Phi_{r,n}(iy)=\frac{1}{2\pi}\int_0^{2\pi}F_y(\theta)e^{ir\psi_y(\theta)}d\theta$,
    with $F_y(\theta)=e^{in\theta+\psi_y(\theta)}$ a smooth function. For such an integral by the method of stationary
    phase
    \[\frac{1}{2\pi}\int_0^{2\pi}F_y(\theta)e^{ir\psi_y(\theta)}d\theta=O(\frac{1}{\sqrt{r}}).\]
    The implied constant, can be given explicitly in terms of
    $\norm{F_y}_\infty$, $\norm{F_y'}_\infty$, $\norm{\psi_y'}$, $\norm{\psi_y''}$ and
    $\psi_y''(\frac{\pi l}{2})$, and hence can be chosen uniformly for any
    bounded segment not containing $1$.
    \end{proof}

\begin{defn}
    For $n\in\bbZ$, let $C^\infty_n(\bbH)$ denote the space of smooth
    compactly supported functions on $\bbH$ satisfying
    $f(kz)=\chi_n(k)f(z)$. Define the $n$-spherical transform on
    $C^\infty_n(\bbH)$ by
    \[\calS_n(f)(r)=\int_\bbH f(z)\Phi_{r,-n}(z)dz.\]
    (For $n=0$, this is also known as the Selberg transform.)
\end{defn}

    We say that a holomorphic function $h(r)$ is of uniform exponential
    type $R$, if $\forall N\in\bbN,\; h(r)\ll_N
    \frac{e^{R|\Im{(r)}|}}{(1+|r|)^N}$. Let $PW(\bbC)$ denote the space
    of holomorphic functions of uniform exponential type, and
    $PW_n(\bbC)$ the subspace of holomorphic functions of uniform
    exponential type satisfying the functional equation
    $P_n(2ir)h(-r)=P_n(-2ir)h(r)$.

\begin{prop}
    The $n$-spherical transform, $\calS_n$,
    is a bijection of $C^\infty_n(\bbH)$ onto $PW_n(\bbC)$,
    with inverse transform given by
    \[\calS_n^{-1}h(z)=\frac{1}{2\pi}\int_0^\infty h(r)\Phi_{-r,n}(z)r\tanh(\pi
    r)dr.\]
    Moreover, if $h\in PW_n(\bbC)$ is of uniform
    exponential type $R$, then $f=\calS_n^{-1}h\in C^\infty_n(\bbH)$ is supported
    in the disc $d(z,i)<R$.
\end{prop}
\begin{proof}
    For any $f\in C^\infty_c(\bbH)$ it's Helgason-Fourier transform
    is given by
    $$\tilde{f}(r,k)=\int_\bbH f(z)\vphi_{-r}(k^{-1}z)dz.$$
    This transform is a bijection of $C_c^\infty(\bbH)$ onto the space of holomorphic functions with uniform exponential type
    satisfying the functional equation
    \[\int_K\vphi_r(k^{-1}z)\tilde{f}(r,k)dk=\int_K \vphi_{-r}(k^{-1}z)\tilde{f}(-r,k)dk.\]
    The inverse transform is given by
    \[f(z)=\frac{1}{2\pi}\int_0^\infty\int_K \tilde{f}(r,k)\vphi_{r}(k^{-1}z)r\tanh(\pi
    r)dr,\]
    and if $\tilde{f}(r,k)$ is of uniform exponential type $R$, then $f(z)$ is supported on $d(i,z)\leq R$ \cite[Theorem 4.2]{Helgason81}.
    The above proposition now follows directly from the identity (verified by a simple
    computation)
    $$\forall f\in C^\infty_n(\bbH),\quad \tilde{f}(r,k)=\chi_n(k)\calS_n f(-r).$$
\end{proof}

    For $f\in C_c^\infty(\bbH)$ let $L[f]\colon C^\infty(G/K)\to
    C^\infty(G)$ be the convolution operator defined by
    \[L[f]u(g)=\int_\bbH f(g^{-1}w)u(w)dw.\]
    \begin{lem}\label{lLf}
    Let $\phi\in C^\infty(\bbH)$ be a Laplacian eigenfunction with eigenvalue $(r^2+\frac{1}{4})$, and let
    $\phi_n\in C^\infty(G)$ satisfy $\phi_{n\pm 1}=\frac{1}{2ir+1\pm
    2n}E^\pm \phi_n$ with $\phi_0(g)=\phi(g(i))$.
    Then for any $f\in C^\infty_n(\bbH)$,
    \[L[f]\phi_0=\calS_nf(r)\phi_{-n}\]
    \end{lem}
    \begin{proof}
    First, by \cite[Theorem 4.3]{Helgason81}, any Laplacian
    eigenfunction $\phi$ can be expressed as an
    integral $\phi(z)=\int_K\vphi_{r}(kz)dT(k)$
    with respect to a suitable distribution on $K$. Hence, it is sufficient to show this in the special case where
    $\phi(z)=\vphi_{r}(k z)$ for arbitrary $k\in K$.
    Next, note that if $\tilde\phi(z)=\phi(kz)$ then $(L[f]\tilde \phi)(g) =L[f]\phi(k g)$ and also $\tilde\phi_n(g)=\phi_n(kg)$
    (because the left action of $K$ commutes with $E^\pm$). Hence it is sufficient to
    show the equality only for $\phi(z)=\vphi_r(z)$.
    Finally, note that the functions $\phi_n(g)=\vphi_r(g(i))\chi_n(g)$
    satisfy the above recursion relation. 
    It thus remains to show that $L[f]\vphi_r(g)=S_nf(r)\vphi_r(g(i))\chi_{-n}(g)$.

    Fix $g=p_zk$, then (after the change of variables $w\mapsto p_zw$)
    \[ L[f] \vphi_r(p_zk)=\int_\bbH f(k^{-1}w)\vphi_r(p_zw)dw.\]
    The function $\vphi_r$ satisfies $\vphi_r(p_zw)=\vphi_r(z)\vphi_r(w)$ so that
    \begin{eqnarray*}
    L[f] \vphi_r(p_zk)&=&\vphi_r(z)\int_\bbH f(w)\vphi_r(k w)dw=\\
    =\vphi_r(z)\tilde{f}(-r,k^{-1})&=&\calS_n f(r)\chi_{-n}(k)\vphi_r(z).
    \end{eqnarray*}
    concluding the proof.
    \end{proof}

\subsection{Quantization}\label{sQuant2}
    We now wish to relate the functionals
    $S_{\phi_ k}$ to functionals obtained by diagonal matrix elements of some quantization procedure.
    For this we use a generalization of Zeldich's quantization
    procedure via Helgasons Fourier transform \cite{Zelditch87}.
    For any smooth function $a\in C^\infty(TX)$ we assign its
    quantization which is an integral operator $\Op(a)$ acting on $L^2(X)$. Recall the
    map $(z,\xi)\mapsto (p_zk,E(z,\xi))$ from $TX$ to $\Gamma\bs\calG\times{[0,\infty)}^d$ and think of
    a function on $TX$ as a function $a=a(p_zk,r)$ with the parametrization $r_j=\sqrt{E_j-\frac{1}{4}}$. Let
    \[\tilde{a}(z,w)=\frac{1}{(2\pi)^d}\int_{{\bbR^+}^d}\int_\calK
    a(p_zk,r)\left(\prod_{j=1}^d \vphi_{r_j}(k_j^{-1}w_j)r_j\tanh(\pi
    r_j)\right)drdk,\]
    be the inverse Helgason-Fourier transform (in all of the
    $(k_j,r_j)$ coordinates). We then define the operator $\Op(a)$ by the kernel
    \[K(z,w)=\tilde{a}(z,p_z^{-1}w).\]
    Since $a$ is $\Gamma$ invariant, this kernel satisfies $K(\gamma z,\gamma
    w)=K(z,w)$, and hence defines an operator on $L^2(X)$.

    In the following lemma, we show that the Wigner distribution $a\mapsto \langle
    \Op(a)\phi_k,\phi_k\rangle_X$ is supported on $\Sigma(\lambda_k)\cong \Gamma\bs\calG$ and
    coincides there with $S_{\phi_k}$.
    \begin{lem}\label{lequiv}
    Let $a=a(g,r)\in C^\infty(TX)$ be
    holomorphic of uniform exponential type (in the $r_j$ variables)
    and satisfy that $\forall j=1,\ldots, d$ the expression
    \[\int_{K_j} a(p_zk,r)
    \vphi_{r_j}(k_j^{-1}w_j)dk_j,\]
    is invariant under the substitution $r_j\mapsto -r_j$. Then
    \[\langle \Op(a)\phi_k,\phi_k\rangle=S_{\phi_k}(a(\cdot,r_k)).\]
    \end{lem}
    \begin{proof}
     We will give the proof in the special case where the function
     $a$ is of the form $a(g,r)=a(g)h(r)$ with $h(r)=\prod_j
     h_j(r_j)$ and $a\in \calF_n(Y)$ is of some fixed $\calK$-type $n\in\bbZ^d$.
     (This is the only case that we will use, however, the general
     statement can be deduced by decomposing $a(g,r)$ into its $\calK$-Fourier series.)
     For $a$ of the above type, the functional equation is equivalent to the requirement
     that the functions $h_j\in PW_{n_j}(\bbC)$. We can now write the kernel as
     \[K_a(z,w)=a(p_z)\prod_{j=1}^d f_j(p_{z_j}^{ -1}w_j),\]
     with $f_j=\calS_{n_j}^{-1}(h_j)\in C_{n_j}^\infty(\bbH)$. In
     particular the operator $\Op(a)$ is given by a tensor product of convolution operators
    $$\Op(a)\phi(z)=a(p_z)L[f_1]\otimes\cdots \otimes L[f_d]\phi(z).$$
    Since $\phi_k(z)$ are joint eigenfunctions of all partial Laplacians, Lemma \ref{lLf} (applied separately to each coordinate) implies
    $$\Op(a)\phi_k(z)=a(p_z)h(r_k)\phi_{k,-n}(p_z).$$
    We thus get that
    \begin{eqnarray*}
    \langle
    \Op(a)\phi_k,\phi_k\rangle_X &=& \int_X \Op(a)\phi_k(z)\overline{\phi_k(z)}dz\\
    &=&\int_X
    a(p_z)h(r_k)\phi_{k,-n}(p_z)\overline{\phi_k(z)}dz\\
    &=&\int_Y
    a(g)h(r_k)\phi_{k,-n}(g)\overline{\phi_{k,0}(g)}dg\\
    &=&h(r_k)\langle a
    \phi_{k,-n},\phi_{k,0}\rangle_Y=S_{\phi_k}(a(\cdot,r_k)).
    \end{eqnarray*}
    \end{proof}

\section{A Local Weyl's Law}\label{sLW}
We now give the proof of Theorem \ref{tWL}, showing that for large eigenvalues, 
on average, the distributions $S_{\phi_k}$
defined above converge to the volume measure of $Y$. 

\subsection{A Trace Formula}
    The main ingredient in the proof will be a trace formula, relating the
    sum over the eigenvalues to a sum over conjugacy classes in
    $\Gamma$.
    Recall the setting: $X=\Gamma\bs \calH$, $\{\phi_k\}\in C^\infty(X)$ is an orthonormal
    basis for $L^2(X)$ composed of joint Laplacian eigenfunctions
    (with eigenvalues $\lambda_{k,j}=(r_{k,j}^2+\frac{1}{4})$ respectively) and $S_{\phi_k}$ the corresponding distributions.

    For any $1\leq j\leq d$, fix $f_j\in C_{n_j}^\infty(\bbH)$, and let
    $h_j(r_j)=\calS_{n_j}f_j\in PW_{n_j}(\bbC)$ be the corresponding spherical transforms.
    Denote by $h(r)=\prod_j h_j(r_j)$ and by
    $f(z)=\prod_j f_j(z_j)$. For any $\gamma\in\Gamma$, let $\Gamma_\gamma$
    be the centralizer of $\gamma$ in $\Gamma$ and let $\calF_\gamma\subseteq\calH$
    be a fundamental domain for $\Gamma_\gamma$.
    \begin{thm}\label{tTRACE}
    For any observable $a\in \calF_n(X)$
    \[\sum_k h(r_k)S_{\phi_k}(a)=\sum_{\{\gamma\}}\int_{\calF_\gamma} a(p_z)f(p_z^{-1}\gamma
    z)dz,\]
    where the right hand sum is over the conjugacy classes in $\Gamma$.
    \end{thm}
    \begin{rem}
    In the special case, when $n=0$ and $a\equiv 1$ is the constant
    function, the terms
    $\int_{\calF_\gamma}f(p_{z}^{-1}\gamma
    z)dz$ can be computed explicitly in terms of the Fourier transform of
    $h$, retrieving the Selberg trace formula.
    \end{rem}
    \begin{proof}
    Consider the operator $\Op(ah)$ given by the kernel
    $$K(z,w)=a(p_z)f(p_{z}^{ -1}w).$$
    We can think of $\Op(ah)$ as an operator on
    $L^2(\Gamma\bs\bbH)$ with kernel given by
    $$K_\Gamma(z,w)=\sum_\gamma K(z,\gamma w).$$
    Write the trace of this operator in two different ways.
    First, since $\phi_k$ is an orthonormal basis for
    $L^2(X)$, by Lemma \ref{lequiv}
    \[\Tr(\Op(ah))=\sum_k \langle
    \Op(a)\phi_k,\phi_k\rangle_X=\sum_k h(r_k)S_{\phi_k}(a).\]

    On the other hand, if $\calF\subseteq\calH$ is a fundamental domain for
    $\Gamma$ then
    \[\Tr(\Op(ah))=\int_\calF K_\Gamma(z,z)dz=\sum_{\gamma}\int_\calF  K(z,\gamma z)dz.\]
    Note that if $\gamma'=g^{-1}\gamma g$ are conjugated in $\Gamma$
    then
    \[\int_\calF  K(z,\gamma' z)dz=\int_\calF  K(gz,\gamma gz)dz=\int_{g\calF}  K(z,\gamma
    z)dz.\]
    We can thus write
    \begin{eqnarray*}
    \Tr(\Op(ah))&=&\sum_{\{\gamma\}}\int_{\calF_\gamma}K(z,\gamma
    z)dz\\
    &=&\sum_{\{\gamma\}}\int_{\calF_\gamma}a(p_z)f(p_{z}^{-1}\gamma
    z)dz
    \end{eqnarray*}
    where $\calF_\gamma=\sum_{g\in
    \Gamma/\Gamma_\gamma}g\calF_\gamma$ is the fundamental domain
    for $\Gamma_\gamma$.
    \end{proof}

\subsection{Smoothing}
    In order to use the trace formula to evaluate the sum
    $\sum_{k\in \calI(\bfL)}S_{\phi_k}(a)$, we need to approximate
    the window function by a smooth function admissible in the trace
    formula.
    \begin{defn}
        We say that a smooth function $h\in C^\infty(\bbR)$ is
        $\delta$-approximating the window function around $L\in[\frac{1}{2},\infty)$, if it satisfies for real $x>0$
        \[|h(x)-\id_{[L-\frac{1}{2},L+\frac{1}{2}]}(x)|=
        \left\lbrace\begin{array}{cc}
        O(\delta) & |x-L|\leq \frac{1}{2}-\sqrt{\delta}\\
        O(1) &  \frac{1}{2}-\sqrt{\delta}\leq |x-L|\leq
        \frac{1}{2}\\
        O_N(\delta(\frac{1}{|x-L|-1/2})^N) &
        |x-L|>\frac{1}{2}\\
        \end{array}\right.\]
        where $\id_{[\alpha,\beta]}$ is the indicator function of $[\alpha,\beta]$.
    \end{defn}
    Let $\Theta(r)=\prod_j \id_{[-\frac{1}{2},\frac{1}{2}]}(r_j)$ denote a
    window function around zero in $\bbR^d$. If $h_{L_j,\delta}$ are functions $\delta$-approximating the window functions around
    $L_j$ respectively, then their product
    $h_{\bfL,\delta}(r)=\prod_jh_{L_j,\delta}(r_j)$ is a good
    approximation to the window function $\Theta(r-\bfL)$
    around $\bfL=(L_1,\ldots L_j)$ in the following sense.
    \begin{prop}\label{pSmooth}
    \[\limsup_{\norm{\bfL}\to\infty}\frac{1}{L_1\cdots L_d}\sum_{r_{k}\in\bbR^d} |h_{\bfL,\delta}(r_k)-\Theta(r_{k}-\bfL)|=O(\sqrt{\delta})\]
    \end{prop}
    \begin{proof}
    Appendix \ref{scount}, Proposition \ref{pSmooth2}.
    \end{proof}

    For $n\in\bbZ$, recall that $PW_{n}(\bbC)$ is the space of
    holomorphic functions $h(x)$ of uniform exponential type, satisfying the
    functional equation $P_{n}(-2i)h(x)=P_n(2ix)h(-x)$ with
    $$P_n(x)=(\frac{x+1}{2})(\frac{x+1}{2}+1)\cdots(\frac{x+1}{2}+|n|-1).$$
    We will show
    that for any fixed $n\in\bbZ$, there are functions in $PW_n(\bbC)$ that
    $\delta$-approximate the window functions. For this we need
    the following lemma.
    \begin{lem}\label{lCORRECT}
    For fixed $n\in\bbZ$, there are holomorphic functions $F_\delta(x)$ satisfying
        \begin{itemize}
        \item The Fourier transform $\hat{F}_\delta\in C^\infty_c(\bbR)$ is compactly supported.
        \item $\forall |m|\leq
        |n|,\;F_\delta(\frac{im}{2})=1$
        \item $F_\delta(x)=O(\delta)$, uniformly for real $x\in\bbR$.
        \end{itemize}
    \end{lem}
    \begin{proof}
    For $0<\delta<1$, let $G_\delta(x)=\sin(x/\delta)\prod_{1\leq |m|\leq
    |n|}\frac{2x-im}{x/\delta-m\pi}$. Then $G_\delta(\frac{im}{2})=0$, the derivative
    $G'_\delta(\frac{im}{2})\gg e^{\frac{m}{\delta}}\gg
    \frac{1}{\delta}$,
    and for real $x$ the function $|G_\delta(x)|\leq 1$ is
    bounded. The function,
    \[F_\delta(x)=\sum_{1\leq|m|\leq |n|}
    \frac{G_\delta(x)}{G_\delta'(im)(x-im)},\]
   then satisfies the above properties \footnote{I thank Mikhail Sodin for showing me this construction.}.
    \end{proof}

    \begin{prop}\label{pSmooth1}
        For fixed $n\in\bbZ$, for any $L\geq\frac{1}{2}$ and $\delta>0$ there is a function $h_{L,\delta}\in
        PW_n(\bbC)$ (with exponential type depending on $\delta$ but not on
        $L$), that is $\delta$-approximating the window function around $L$.
    \end{prop}
    \begin{proof}
        Fix a positive even holomorphic function $\rho\in PW_0(\bbC)$ with Fourier transform $\hat\rho$ supported in $[-1,1]$
        and $\hat\rho(0)=1$. For any $\delta>0$, let
        $\rho_\delta(x)=\frac{1}{\delta}\rho(\frac{x}{\delta})$ and
        define a smoothed window function by convolution with the window
        function $\id_{[-\frac{1}{2},\frac{1}{2}]}$.
        Then the smoothed function $\id_\delta=\rho_\delta*\id_{[-\frac{1}{2},\frac{1}{2}]}$
        satisfies
        \[\id_\delta(x)=\left\lbrace\begin{array}{cc}
        1+O(\delta) &|x|\leq 1/2-\sqrt{\delta}\\
        O((\frac{\delta}{|x|-1/2})^N)& |x|>\frac{1}{2}\\
        \end{array}\right.\]
        We now want to deform the smoothed function
        $\id_\delta(x-L)$ into a function in $PW_n(\bbC)$. For $n=0$ we can simply take
        \[h_{L,\delta}(x)=\id_\delta(-x-L)+\id_\delta(x-L)\in PW_0(\bbC).\]

        Otherwise, let
        $F_\delta(x)$ be as in Lemma \ref{lCORRECT}, and define the function
        \[h_{L,\delta}(x)=(1-F_\delta(x))\id_\delta(x-L)+\frac{P_n(2ix)}{P_n(-2ix)}(1-F_\delta(-x))\id_\delta(-x-L).\]
        The function $h_{L,\delta}$ obviously satisfies the functional
        equation.
        The zeros of $1-F_\delta(x)$ cancel the poles of
        $\frac{P_n(2ix)}{P_n(-2ix)}$, and since the Fourier
        transform $\hat{F}_{\delta}$ is compactly supported, $h_{L,\delta}$ is of uniform exponential type
        (depending only on $\delta$).
        It remains to show that it indeed approximates the window
        function.

        First, note that for $x\in\bbR$ the function $|1-F_\delta(x)|=O(1)$ and $|\frac{P_n(-ix)}{P_n(ix)}|=1$,
        so that
        $|h_{L,\delta}(x)|=O(1)$ is bounded. Now, for $|x-L|>1/2$ we can
        bound
        \[|h_{L,\delta}(x)|\leq
        |(1-F_\delta(x))||\id_\delta(x-L)|+|(1-F_\delta(-x))||\id_\delta(-x-L)|.\]
        The function $|1-F_\delta(x)|$ is bounded and
        $\id_\delta(\pm x-L)\ll(\frac{\delta}{|x-L|-\frac{1}{2}})^N$,
        hence
        $|h_{L,\delta}(x)|\ll(\frac{\delta}{|x-L|-\frac{1}{2}})^N$.

        Next, for $|x-L|\leq 1/2-\sqrt{\delta}$ we can bound
        $|h_{L,\delta}(x)-1|$ by
        \[|\id_\delta(x-L)-1|+
        |F_\delta(x)||\id_\delta(x-L)|+|(1-F_\delta(-x))||\id_\delta(-x-L)|.\]
        The first term is bounded by $O(\delta)$, the second term
        is bounded by $O(\delta)$ (because
        $F_\delta(x)=O(\delta)$), and the last
        term is also bounded by $O(\delta)$ (since $|x+L|\geq \frac{1}{2}+\sqrt{\delta}$).
    \end{proof}

\subsection{Proof of Theorem \ref{tWL}}
    Let $a\in C^\infty(Y)$. With out loss of generality we can assume that
    $\int_Y a=0$ and that $a$ is of some fixed $\calK$-type $n$. We thus need
    to show that
    \[\lim_{\norm{\bfL}\to\infty}\frac{1}{N(L)}\sum_kS_{\phi_k}(a)=0.\]

    For $\delta>0$ and $j=1,\ldots, d$ let $h_{L_j,\delta}\in
    PW_{n_j}(\bbC)$ (with exponential type depending on $\delta$ but not on $L_j$)
    $\delta$-approximate the window function around
    $L_j$, and let $h_{\bfL,\delta}(r)=\prod h_{L_j,\delta}(r_j)$.

    Use the trace formula for $h_{\bfL,\delta}(r)$ to get
    \[\sum_k h_{\bfL,\delta}(r_k)S_{\phi_k}(a)=\sum_{\{\gamma\}}\int_{\calF_\gamma} a(p_z)(\prod_{j=1}^d
    f_{L_j,\delta}(p_{z_j}^{-1}\gamma_j
    z_j))dz,\]
    where $f_{L_j,\delta}=\calS_{n_j}^{-1}h_{L_j,\delta}\in C_{n_j}^\infty(\bbH)$ and $\calF_\gamma\subseteq\bbH\times\cdots\times\bbH$
    is the fundamental domain for $\Gamma_\gamma$.

First notice that the conjugacy class of the identity does not contribute anything. To see this write its contribution as
\[\int_{\calF} a(p_z)(\prod_{j=1}^d f_{L_j,\delta}(p_{z_j}^{-1}z_j))dz=(\prod_{j=1}^df_{L_j,\delta}(i))\int_{\calF}a(p_z)dz.\]
If there is some $n_j\neq 0$ then $f_{L_j,\delta}(i)=0$. Otherwise $a$ is $\calK$ invariant and $\int_{\calF}a(p_z)dz=\int_{Y}a(g)dg=0$.

Next, recall that the functions $f_{L_j,\delta}$ are compactly supported
so we can replace the noncompact domains $\calF_\gamma$ by compact
domains of the form
$\tilde\calF_\gamma=\set{z\in\calF_\gamma\colon
d(z_j,\gamma_jz_j)<M}$ for some constant $M=M(\delta)$ depending
on $\delta$. Denote by $l_{\gamma_j}=\inf_\bbH d(z_j,\gamma_jz_j)$ and note that there can be only a
finite number of conjugacy classes satisfying
that $\max_j l_{\gamma_j}\leq M$, hence, there are
only a finite number of conjugacy class that contribute to the
sum (the number depending again on $\delta$ but
not on $\bfL$).

We now use the inverse transform to estimate the size of
$f_{L_j,\delta}$,
\[f_{L_j,\delta}(p_{z_j}^{-1}\gamma_j
z_j)=\frac{1}{2\pi}\int_0^\infty h_{L_j,\delta}(r)\Phi_{r,n_j}(p_{z_j}^{-1}\gamma_j z_j)
r_j\tanh(\pi
r_j)dr_j.\]

Since we assume $\Gamma$ is irreducible and co-compact for any nontrivial conjugacy classes $\{\gamma\}$, we know that $\gamma_j$ is either hyperbolic or elliptic. If $\gamma_j$ is hyperbolic we can use Lemma
\ref{lBoundSPherical} to bound
$\Phi_{r_j,n_j}(p_{z_j}^{-1}\gamma_j
z_j)\ll_\delta \frac{1}{\sqrt{r_j}}$ uniformly in the annulus  $l_{\gamma_j}\leq
d(p_{z_j}^{-1}\gamma z_j,i)\leq M$.
We thus get the bound
\[f_{L_j,\delta}(p_{z_j}^{-1}\gamma_j z_j)\ll_\delta \int_0^\infty |h_{L_j,\delta}(r_j)|\sqrt{r_j}dr_j\ll_\delta\sqrt{L_j}.\]

In the case where $\gamma_j$ is elliptic, for any $\epsilon>0$ as before
we can bound $f_{L_j,\delta}(p_{z_j}^{-1}\gamma_j z_j)\ll_{\epsilon,\delta}\sqrt{L_j}$, uniformly  in
the annulus $\epsilon\leq
d(p_{z_j}^{-1}\gamma z_j,i)\leq M$. For $d(p_{z_j}^{-1}\gamma
z_j,i)<\epsilon$ (i.e., in an $\epsilon$-neighborhood of the fixed
point of $\gamma_j$) we use the trivial bound $f_{L_j,\delta}(p_{z_j}^{-1}\gamma_j
z_j)\ll L_j$ (coming from the estimate $\phi_{r,n_j}(p_{z_j}^{-1}\gamma_j
z_j)=O(1)$).

Plugging these estimates in the integral, for strictly hyperbolic
conjugacy classes
\[\int_{\tilde\calF_\gamma} a(p_z)(\prod_{j=1}^d
f_{L_j,\delta}(p_{z_j}^{-1}\gamma_j
z_j))dz=O_{\delta}(\sqrt{L_1\cdots L_d}),\]
and for mixed conjugacy classes (where some of the elements are
elliptic)
\[\int_{\tilde\calF_\gamma} a(p_z)(\prod_{j=1}^d
f_{L_j,\delta}(p_{z_j}^{-1}\gamma_j
z_j))dz=O_{\delta,\epsilon}(\sqrt{L_1\cdots L_d})+O(\epsilon L_1\cdots L_d).\]
This is true for any $\epsilon>0$, hence for any
conjugacy class
\[\int_{\tilde\calF_\gamma} a(p_z)(\prod_{j=1}^d
f_{L_j,\delta}(p_{z_j}^{-1}\gamma_j
z_j))dz=o(L_1\cdots L_d),\]
and thus for the whole sum
\[\sum_k  h_{\bfL,\delta}(r_k)S_{\phi_k}(a)=o(L_1\cdots L_d).\]
Taking the limit, recalling that $N(\bfL)\gg L_1\cdots L_d$ (Proposition \ref{pCOUNT2}) we get that
\[\lim_{\norm{\bfL}\to\infty}\frac{1}{N(\bfL)}\sum_k h_{\bfL,\delta}(r_k)S_{\phi_k}(a)=0.\]

The contribution from the exceptional eigenfunctions, where $r_{k,j}$ is imaginary, is
negligible (see Lemma \ref{lSmall1}), hence
\[\lim_{\norm{\bfL}\to\infty}\frac{1}{N(\bfL)}\sum_{r_{k}\in\bbR^d} h_{\bfL,\delta}(r_k)S_{\phi_k}(a)=0.\]
Because $h_{L_j,\delta}(r_j)$ are $\delta$-approximating the window functions around
$L_j$, by Proposition \ref{pSmooth}
\[\limsup_{\norm{\bfL}\to\infty}\frac{1}{N(\bfL)}\sum_{r_{k}\in\bbR^d}
|h_{\bfL,\delta}(r_{k})-\Theta(r_{k}-\bfL)|=O(\sqrt{\delta}),\]
implying that
\[\limsup_{\norm{\bfL}\to\infty}\frac{1}{N(\bfL)}\sum_{k\in\calI(L)} S_{\phi_k}(a)=O(\sqrt\delta).\]
This is true for any $\delta>0$, hence
\[\lim_{\norm{\bfL}\to\infty}\frac{1}{N(\bfL)}\sum_k S_{\phi_k}(a)=0.\]

\appendix
\setcounter{thm}{0}
\renewcommand{\thethm}{\Alph{thm}}
\section{Counting Eigenvalues}\label{scount}

    Let $X=\Gamma\bs\calH$ be a compact locally symmetric
    space with $\calH=\bbH\times\cdots\times\bbH$ a product of $d$ hyperbolic
    planes, $\G=\PSL(2,\bbR)^d$ the group of isometries,
    and $\Gamma\subseteq \G$ an irreducible
    co-compact lattice. Let $\{\phi_k\}$ be a basis for
    $L^2(X)$ composed of Laplacian eigenfunctions  (with eigenvalues $\lambda_{k,j}=\frac{1}{4}+r_{k,j}^2$).
    For $\bfL=(L_1,\ldots,L_d)\in[\frac{1}{2},\infty)^d$ let
    $$N(\bfL)=\sharp\set{k\colon \norm{r_{k}-\bfL}_\infty \leq \frac{1}{2}}.$$

    \begin{thm}\label{tCOUNT}
    As $\bfL\to\infty$,
      \[L_1\cdots L_d\ll N(\bfL)\ll L_1\cdots L_d\]
    \end{thm}
      \begin{rem}
      This theorem can be deduced from the analysis of
      Duistermaat, Kolk and Varadajan on the spectrum of compact locally symmetric spaces
      \cite[Theorem 7.3]{DuisteremaatKolkVaradarajan79}.
      However, for the sake of completeness, we will include here a self contained proof of this result.
      \end{rem}

    In order to prove Theorem \ref{tCOUNT}, we will prove separately the upper and lower
    bounds. For the upper bound, we consider the number of eigenvalues in a scaled
    window $N(\bfL,\epsilon)=\sharp\set{k\colon \norm{r_{k}-\bfL}_\infty\leq
    \frac{\epsilon}{2}}$.

    \begin{prop}\label{pCOUNT1}
    There is a constant $c_1>0$ such that for every $\epsilon>0$
    \[\limsup_{\norm{\bfL}\to\infty}\frac{N(\bfL,\epsilon)}{L_1\cdots L_d}\leq c_1\epsilon^{d}\]
    \end{prop}
    In particular for $\epsilon=1$, $N(\bfL)\ll L_1\cdots L_d$. Now
    for the lower bound:
    \begin{prop}\label{pCOUNT2}
    There is a constant $c_2>0$, such that
    \[\liminf_{\norm{\bfL}\to\infty}\frac{N(\bfL)}{L_1\cdots L_d}\geq c_2\]
    \end{prop}
    \subsection{Selberg Trace Formula}
    The main tool we use for the proof of Propositions \ref{pCOUNT1} and \ref{pCOUNT2} is the Selberg trace formula
    (see \cite[Sections 1-6]{Efrat87} for the full derivation of the trace formula in this setting).

    For any $\gamma\in\Gamma$ denote
    by $\{\gamma\}\in\Gamma^\sharp$ its conjugacy class, by
    $\Gamma_\gamma$ its centralizer in $\Gamma$, and by
    $\calG_\gamma$ it centralizer in $\calG$.
    Let $h_j(r_j)\in C^\infty(\bbR)$ be even and
    holomorphic in the strip $|\im(r_j)|\leq C$ for some fixed $C>\frac{1}{2}$. For
    any conjugacy class $\{\gamma\}\in\Gamma^\sharp$,
    let $c_\gamma=\vol(\Gamma_\gamma\bs
    G_\gamma)$. Recall that for any $\gamma\in\Gamma$ its
    projections to the different factors are either hyperbolic, $\gamma_j\sim \begin{pmatrix} e^{l_j/2} & 0\\ 0 &
    e^{-l_j/2}\end{pmatrix}$, or elliptic $\gamma_j\sim \begin{pmatrix} \cos\theta_j & \sin\theta_j\\
    -\sin\theta_j &
    \cos\theta_j \end{pmatrix}$.
    Define the functions $\tilde{h}_j(\gamma_j)$ by
    \[\tilde{h}_j(\gamma_j)=\frac{\hat{h}(l_{j})}{\sinh(l_{j}/2)},\]
    when
    $\gamma_j$ is hyperbolic, and
    \[\tilde{h}_j(\gamma_j)=\frac{1}{\sin\theta_j}\int_{-\infty}^{\infty}
    \frac{\cosh[(\pi-2\theta_j)r]}{\cosh(\pi r)}h(r)dr\] when $\gamma_j$ is
    elliptic.
    The Selberg trace formula, applied to the product $h(r)=\prod
    h_j(r_j)$, then takes the form
    \[\sum_k h(r_k)=\prod_j\left(\frac{1}{4\pi}\int_{\bbR} h_j(r_j)r_j\tanh(\pi
    r_j)dr_j\right)+\sum_{\{\gamma\}}c_\gamma
    \tilde{h}(\gamma),\]
    where the right hand sum is over the nontrivial conjugacy classes
    $\{\gamma\}\in\Gamma^\sharp$ and $\tilde{h}(\gamma)=\prod
    \tilde{h}_j(\gamma_j)$.

    \subsection{Exceptional eigenfunctions}
    Recall that an exceptional eigenfunction is an eigenfunction for
    which some of the eigenvalues are small $0<\lambda_{k,j}<\frac{1}{4}$ (or equivalently $r_{k_j}\in i(0,\frac{1}{2})$).
    We now do a separate treatment of the contribution of these eigenfunctions to the trace formula.
    We show that the exceptional eigenfunctions are of density zero,
    so that their contribution to the trace formula can be neglected.

    For any subset $\calJ\subset\set{1,\ldots,d}$, denote by
    $\calI(\calJ)$ the set of (exceptional) eigenfunctions $\phi_k$ for which the
    $j$'th partial eigenvalue is small for $j\in \calJ$ (and not small otherwise).
    That is
    \[\calI(\calJ)=\set{k\colon \forall j\in\calJ, \lambda_{k_j}< \frac{1}{4},\;\forall j\not\in\calJ, \lambda_{k,j}\geq \frac{1}{4}}.\]
    Also denote by
    \[\calI(\calJ,\bfL)=\set{k\colon \forall j\in\calJ, \lambda_{k_j}< \frac{1}{4},
    \;\forall j\not\in\calJ, |r_{k,j}-L_j|\leq \frac{1}{2}},\]
    and let $N(\calJ,\bfL)=\sharp\calI(\calJ,\bfL)$.

    \begin{lem}\label{lSmall1}
    For any nonempty subset $\calJ\subset\{1,\ldots,d\}$,
    \[\lim_{\norm{\bfL}\to\infty}\frac{N(\calJ,\bfL)}{L_1\cdots
    L_d}=0\]
    \end{lem}
    \begin{proof}
    We will prove this for $\calJ=\set{1,\ldots,s-1}$ (the proof is
    analogous for any other subset).
    For any $T>0$ define the
    function
    \[h_{T,\bfL}(r)=e^{-\frac{T}{2}(r_1^2+\cdots+r_{s-1}^2)}\prod_{j=s}^d(e^{-\frac{(r_j-L_j)^2}{2}}+e^{-\frac{(r_j+L_j)^2}{2}})^2.\]
    When the coordinates $r_j$ are real or
    imaginary, the function $h_{T,\bfL}(r)$ is a positive real function. Moreover, if we assume that $r_j$ is imaginary
    for $1\leq j\leq s-1$, and that $|r_j-L_j|\leq \frac{1}{2}$ for $s\leq
    j\leq d$, then
    $h_{T,\bfL}(r)>e^{-\frac{d-s+1}{2}}\geq\frac{1}{e}$ is uniformly bounded away from zero.
    We can thus bound
    \[N(\calJ,\bfL)\ll\sum_{k}h_{T,\bfL}(r_k).\]
    Now, plugging the functions $h_{T,\bfL}$ in the Selberg trace
    formula we get
    \[N(\calJ,\bfL)\ll\int_{\bbR^d}h_{T,\bfL}(r)\prod_j r_j\tanh(\pi
    r_j)dr+\sum_{\{\gamma\}}c_\gamma
    \tilde{h}_{T,\bfL}(\gamma).\]
    The contribution from the nontrivial conjugacy classes is bounded
    by some constant depending on $T$ but not on $\bfL$, while the
    integral is bounded by
    $O(\frac{L_{s}\cdots L_d}{T})$.
    Dividing by $L_1\cdots L_d$ and taking $\bfL\to\infty$ we get
    \[\limsup_{\norm{\bfL}\to\infty} \frac{N(\calJ,\bfL)}{L_1\cdots
    L_d}=O(\frac{1}{T}).\]
    Now take $T\to\infty$ to conclude the proof.
    \end{proof}

    \begin{lem}\label{lSmall2}
    Let $h_j\in C^\infty(\bbR)$ be holomorphic and satisfy
    $|h_j(r_j)|\ll \frac{1}{|r_j|^3}$ uniformly in the strip $|\im(r)|\leq \frac{1}{2}$.
    Let $\calI_s$ denote the set of exceptional eigenfunctions.
    Define $h_\bfL(r)=\prod_jh_j(r_j-L_j)$, then
    \[\lim_{\norm{\bfL}\to\infty}\frac{1}{L_1\cdots
    L_d}\sum_{k\in\calI_s}h_\bfL(r_k)=0\]
    \end{lem}
    \begin{proof}
    It is sufficient prove this when taking the sum over $k\in\calI(\calJ)$ for
    an arbitrary nonempty subset $\calJ\subset\set{1,\ldots d}$.
    We will show this for $\calJ=\set{s+1,\ldots, d}$ (the proof for any other set is analogous).

    We can write the corresponding sum as
    \[\sum_{k\in\calI(\calJ)}h_\bfL(r_k)=\sum_{\bfM\in\bbZ^s}\sum_{k\in\calI(\calJ,\bfL-\bfM)}h_\bfL(r_k),\]
    where we embed $\bfM=(M_1,\ldots,M_s,0,\ldots,0)\subset\bbZ^d$ in the natural way.

    For fixed $k\in\calI(\calJ,\bfL-\bfM)$ and any $j\not\in\calJ$, $|r_{k,j}- L_j|\geq
    |M_j|-\frac{1}{2}$. We can thus deduce that $|h_j(r_{k,j}- L_j)|=O((\frac{1}{M_j^3}))$. For
    $j\in\calJ$, we have $r_{k,j}- L_j=i\tilde{r}_{k,j}- L_j$ with $\tilde{r}_{k,j}<\frac{1}{2}$
    bounded. Consequently,
    $|h_j(r_{k,j}- L_j)|=O(\frac{1}{L_j})=O(1)$ is bounded. We thus have
   \[\frac{1}{L_1\cdots L_d}\sum_{k\in\calI(\calJ)}h_\bfL(r_k)\ll
    \frac{1}{L_1\cdots L_d}\sum_{\bfM\in\bbZ^s}\frac{N(\calJ,\bfL-\bfM)}{\prod_{j\not\in\calJ}\max(M_j^3,1)}.\]

    On the other hand, from the previous lemma, for every
    $\epsilon>0$ there is $R>0$ so that for $\norm{\bfL}>R$,
    $N(\calJ,\bfL)\leq\epsilon L_1\cdots L_d$. Separate the sum into two terms,
    the first a finite sum over the terms $\bfM$ for which $\norm{\bfL-\bfM}\leq
    R$, and the second when $\norm{\bfL-\bfM}>R$.
    The first term is bounded by
     \[\frac{\sharp\set{\bfM\colon\norm{\bfM}\leq R}\cdot \max\set{N(\calJ,\bfM)\colon\norm{\bfM}\leq
    R}}{L_1\cdots L_d}=O_R(\frac{1}{L_1\cdots L_d}),\]
    and the second by
    \begin{eqnarray*}
    \frac{\epsilon}{L_1\cdots L_d}\sum_{\norm{\bfL-\bfM}>
    R}\frac{\prod_j|L_j-M_j|}{\prod_{j}\max(|M_j|^3,1)}&\leq& \epsilon
    \sum_{\bfM}\prod_{j}\frac{1}{\max(M_j^2,1)}\\
    &\leq & \epsilon (1+\sum_{M\neq 0}\frac{1}{M^2})^s.
    \end{eqnarray*}
    Therefore, when taking $\bfL\to\infty$
    \[\limsup_{\norm{\bfL}\to\infty}\frac{1}{L_1\cdots
    L_d}\sum_{k\in\calI_s}h_\bfL(r_k)=O(\epsilon)\]
    and taking $\epsilon\to 0$ concludes the proof.

    \end{proof}

    \subsection{Proof of Proposition \ref{pCOUNT1}}
    Fix a positive even smooth function $h\in C^\infty(\bbR)$,
    with Fourier transform $\hat{h}$ compactly supported. For each
    $L_j\geq\frac{1}{2},\epsilon>0$ let
    $h_j(r_j)=h_{L_j,\epsilon}(r_j)=h(\frac{r_j-L_j}{\epsilon})+h(\frac{-r_j-L_j}{\epsilon})$.

    For $r_j\in\bbR$ real, the function $h_{L_j,\epsilon}$ is a positive function,
    and for $|r_j-L_j|\leq \frac{\epsilon}{2}$ it is uniformly bounded away from $0$. We can thus bound
    \[N(\bfL,\epsilon)\ll\sum_{r_{k}\in\bbR^d} \prod h_{L_j,\epsilon}(r_{k,j}).\]
    From the previous lemma, the contributions of the exceptional eigenfunctions can be
    bounded by $o(L_1\cdots L_d)$ hence
    \[N(\bfL,\epsilon)\ll\sum_{k} \prod h_{L_j,\epsilon}(r_{k,j})+o(L_1\cdots L_d).\]
    For the full sum, by the Selberg trace formula, we get
    \begin{eqnarray*}
    \sum_{k} \prod h_{L_j,\epsilon}(r_{k,j})&=&\prod_{j=1}^d\left(\int_\bbR h_{L_j,\epsilon}(r_j)r_j\tanh(\pi
    r_j)dr_j\right)\\&+& \sum_{\{\gamma\}}c_\gamma\prod_{j=1}^d
    \tilde{h}_{L_j,\epsilon}(\gamma_j).
    \end{eqnarray*}

    Notice that the Fourier transform $\hat{h}_{L_j,\epsilon}(t)=2\epsilon\cos(L_jt)\hat{h}(\epsilon
    t)$, and since we assumed $\hat{h}$ compactly supported, there are only a finite number (depending on $\epsilon$)
    of nontrivial conjugacy classes contributing to the sum. Each contribution is bounded by some constant
    (not depending on $\bfL$), so that
    \[N(\bfL,\epsilon)\ll\prod_{j=1}^d\left(\int_\bbR h_{L_j,\epsilon}(r_j)r_j\tanh(\pi
    r_j)dr_j\right)+o(L_1\cdots L_d)+O_\epsilon(1).\]
   We can estimate the integral
    \[\int_\bbR h_{L_j,\epsilon}(r_j)r_j\tanh(\pi
    r_j)dr_j\ll \int_0^\infty h(\frac{r-L_j}{\epsilon})rdr\ll
    L_j\epsilon,\]
    to get the bound
    \[N(\bfL,\epsilon)\ll \epsilon^d L_1\cdots L_d+o(L_1\cdots L_d)+O_\epsilon(1).\]
    Now divide by $L_1\cdots L_d$ and take $\bfL\to\infty$ to get
    \[\limsup_{\norm{\bfL}\to\infty}\frac{N(\bfL,\epsilon)}{L_1\cdots L_d}\ll\epsilon^d.\]

    \subsection{Smoothing}\label{sSmooth}
    We now approximate the window function by a smoothed function
    admissible in the Selberg trace formula.
    Recall that a smooth function $h\in C^\infty(\bbR)$ is
        $\delta$-approximating the window function around $L\in\bbR$, if it satisfies for real $x>0$
        \[|h(x)-\id_{[L-\frac{1}{2},L+\frac{1}{2}]}(x)|=
        \left\lbrace\begin{array}{cc}
        O(\delta) & |x-L|\leq \frac{1}{2}-\sqrt{\delta}\\
        O(1) &  \frac{1}{2}-\sqrt{\delta}\leq |x-L|\leq
        \frac{1}{2}\\
        O_N(\delta(\frac{1}{|x-L|-1/2})^N) &
        |x-L|>\frac{1}{2}\\
        \end{array}\right.\]

    \begin{prop}\label{pSmooth2}
    Let $h_{L_j,\delta}\in C^\infty(\bbR)$ be $\delta$-approximating the window functions around
    $L_j$ respectively. Let $h_{\bfL,\delta}(r)=\prod_j h_{L_j,\delta}(r_j)$ be the corresponding approximation
    of the window function $\Theta(r-\bfL)$ around $\bfL$. Then
    \[\limsup_{\norm{\bfL}\to\infty}\frac{1}{L_1\cdots L_d}\sum_{r_{k}\in\bbR^d} |h_{\bfL,\delta}(r_{k})-\Theta(r_{k}-\bfL)|=O(\sqrt{\delta})\]
    \end{prop}
    \begin{proof}
    We can write the sum differently as
    \[\sum_{
    0\neq \bfM\in\bbZ^{d}}\sum_{k\in\calI(\bfL-\bfM)}|\prod_{j=1}^d h_{L_j,\delta}(r_{k,j})|+
    \sum_{k\in\calI(\bfL)}|\prod_{j=1}^d h_{L_j,\delta}(r_{k,j})-1|,\]

    In the first sum, for $k\in\calI(\bfL-\bfM)$, we can bound
    $h_{L_j,\delta}(r_{k,j})=O_N(\delta(M_j)^{-N})$ if $M_j\neq 0$ and $h_{L_j,\delta}(r_{k,j})=O(1)$ otherwise.
    We then evaluate $\sharp\calI(\bfL-\bfM)=O((L_1-M_1)\cdots (L_d-M_d))$ (Proposition \ref{pCOUNT1}) and get a bound
    on the first sum of order
   \[\sum_{\bfM\neq 0}(\prod_{j=1}^d \min(\delta\frac{L_j-M_j}{M_j^N},1))=O(\delta L_1\cdots L_d).\]

   We now evaluate the second sum. For any $\epsilon>0$ denote by
   \[\calI(L,\epsilon)=\set{k\colon \norm{r_{k}-\bfL}_\infty\leq
   \frac{\epsilon}{2}}.\]
   We can separate the sum over $\calI(\bfL)$ to a sum over
   $\calI(\bfL,1-\sqrt{\delta})$ and the rest. For
   $k\in\calI(\bfL,1-\sqrt{\delta})$ we can evaluate
   $h_{L_j,\delta}=1+O(\delta)$, and the number of such eigenvalues
   is bounded by $N(\bfL)=O(L_1\cdots L_d)$ implying that
   \[\sum_{k\in\calI(\bfL,1-\sqrt{\delta})}(\prod_{j=1}^d h_{L_j,\delta}(r_{k,j})-1)=O(\delta L_1\cdots L_d).\]

   We are left with the sum over $\calI(\bfL)\setminus\calI(\bfL,1-\sqrt{\delta})$.
   This set can be covered by $O(\delta^{-\frac{d-1}{2}})$ boxes of size $\delta^{\frac{d}{2}}$.
   Since the number of eigenvalues
   in each such box is bounded by $O(\delta^{\frac{d}{2}}L_1\cdots L_d)+o(L_1\cdots L_d)$ (Proposition \ref{pCOUNT1}) we can bound
   \[\sharp(\calI(\bfL)\setminus\calI(\bfL,1-\sqrt{\delta}))=O(\sqrt{\delta} L_1\cdots L_d)+o(L_1\cdots L_d).\]
    Since the functions $h_{L_j,\delta}=O(1)$ are bounded, this is also the bound for the remaining sum.

   We have thus seen that the difference
   \[\sum_{r_{k}\in\bbR^d} |h_{\bfL,\delta}(r_{k})-\Theta(r_{k}-\bfL)|=O(\sqrt{\delta} L_1\cdots L_d)+o(L_1\cdots L_d).\]
    Dividing by $L_1\cdots L_d$ and taking $\norm{\bfL}\to\infty$ concludes the proof.
    \end{proof}

    \subsection{Proof of Proposition \ref{pCOUNT2}}
    Fix a positive even holomorphic function $\rho\in PW_0(\bbC)$ with Fourier transform $\hat\rho$ supported in $[-1,1]$
    and $\hat\rho(0)=1$. For any $\delta>0$, let
    $\rho_\delta(x)=\frac{1}{\delta}\rho(\frac{x}{\delta})$ and
    define a smoothed window function by convolution
    $\id_\delta=\rho_\delta*\id_{[-\frac{1}{2},\frac{1}{2}]}$.
    For $j=1,\ldots, d$ the function
    \[h_{L_j,\delta}(r_j)=\id_\delta(-r_j-L_j)+\id_\delta(r_j-L_j),\]
    is $\delta$-approximating the window function around $L_j$, and the function
    $h_{\bfL,\delta}(r)=\prod_j h_j(L_j,\delta)(r_j)$ is
    admissible in the Selberg trace formula.
    Hence,
    \[\sum_k h_{\bfL,\delta}(r_{k})=\prod_{j=1}^d\left(\int_\bbR h_{L_j,\delta}(r_j)r_j\tanh(\pi
    r_j)dr_j\right)+\sum_{\{\gamma\}}c_\gamma
    \tilde{h}_{\bfL,\delta}(\gamma).\]
    As in the proof of Proposition \ref{pCOUNT1}, the contribution of the nontrivial conjugacy classes
    is bounded by $O_\delta(1)$. We can bound the integrals
     \[\int_\bbR h_{L_j,\delta}(r_j)r_j\tanh(\pi
    r_j)dr_j\gg L_j,\]
    uniformly for $L_j\geq \frac{1}{2}$. Therefore, there is $c>0$
    such that
    \[\sum_k h_{\bfL,\delta}(r_{k})\geq cL_1\cdots L_d +O_\delta(1).\]
    The contribution of the exceptional eigenfunctions is $o(L_1\cdots L_d)$,
    and by Proposition \ref{pSmooth2} the contribution of all other eigenvalues differ
    from $N(L)$ by $O(\sqrt{\delta} L_1\cdots L_d)+o(L_1\cdots L_d)$.
    We can deduce that

    \[\frac{N(L)}{L_1\cdots L_d}\geq
    c+O(\sqrt{\delta})+O_\delta(\frac{1}{L_1\cdots L_d})+o(1)\]
    Taking $\bfL\to\infty$, and then $\delta\to 0$ concludes the proof.

\bibliographystyle{amsplain}
\def\cprime{$'$}
\providecommand{\bysame}{\leavevmode\hbox
to3em{\hrulefill}\thinspace}
\providecommand{\MR}{\relax\ifhmode\unskip\space\fi MR }
\providecommand{\MRhref}[2]{%
  \href{http://www.ams.org/mathscinet-getitem?mr=#1}{#2}
} \providecommand{\href}[2]{#2}

\end{document}